\documentclass{article}
\usepackage{ijcai23}

\usepackage{times}
\usepackage{soul}
\usepackage{url}
\usepackage[hidelinks]{hyperref}
\usepackage[utf8]{inputenc}
\usepackage[small]{caption}
\usepackage{graphicx}
\usepackage{amsmath}
\usepackage{amsthm}
\usepackage{booktabs}
\usepackage{algorithm}
\usepackage{algorithmic}
\usepackage[switch]{lineno}

\newtheorem{theorem}{Theorem}

\usepackage{graphicx}
\usepackage{amsmath}
\usepackage{amsthm}
\usepackage{booktabs}
\usepackage{algorithm}
\usepackage{algorithmic}

\newtheorem{lemma}{Lemma}

\usepackage[toc,page]{appendix}
\usepackage{hyperref}
\usepackage{cleveref}

\usepackage[toc,page]{appendix}

\usepackage{xcolor}

\usepackage{amsmath,amsthm,bbm,amssymb}
\DeclareMathOperator*{\argmin}{argmin}
\DeclareMathOperator*{\argmax}{argmax}

\newcommand{\br}{\mathit{br}}

\theoremstyle{definition}
\newtheorem*{definition*}{Definition}

\usepackage{booktabs}
\usepackage{subfig}

\newcommand{\citet}[1]{\citeauthor{#1} \shortcite{#1}} 
\newcommand{\citep}{\cite} 

\newcommand{\newterm}[1]{\textbf{\textit{#1}}}

\providecommand{\norm}[1]{\lVert#1\rVert}

\title{Regularization for Strategy Exploration in Empirical Game-Theoretic Analysis}

\author{
Yongzhao Wang
\and
Michael P. Wellman
\affiliations
University of Michigan
\emails
\{wangyzh, wellman\}@umich.com
}

\begin{document}

\maketitle

\begin{abstract}
In iterative approaches to empirical game-theoretic analysis (EGTA), the strategy space is expanded incrementally based on analysis of intermediate game models. 
A common approach to \textit{strategy exploration}, represented by the double oracle algorithm, is to add strategies that best-respond to a current equilibrium. 
This approach may suffer from overfitting and other limitations, leading the developers of the policy-space response oracle (PSRO) framework for iterative EGTA to generalize the target of best response, employing what they term \textit{meta-strategy solvers} (MSSs). 
Noting that many MSSs can be viewed as perturbed or approximated versions of Nash equilibrium, we adopt an explicit regularization perspective to the specification and analysis of MSSs.
We propose a novel MSS called \textit{regularized replicator dynamics} (RRD), which simply truncates the process based on a regret criterion. 
We show that RRD outperforms existing MSSs in various games.
We extend our study to three-player games, for which the payoff matrix is cubic in the number of strategies and so exhaustively evaluating profiles may not be feasible. 
We employ a profile search method that can identify solutions from incomplete models, and combine this with iterative model construction using a regularized MSS\@. 
Finally, we find through experiments that the regret of best-response targets is a strong indicator for the performance of strategy exploration, which provides an explanation for the effectiveness of regularization in PSRO.
\end{abstract}

\section{Introduction}

The methodology of \newterm{empirical game-theoretic analysis} (EGTA) \citep{TuylsPLHELSG20,Wellman16putting} provides a broad toolbox of techniques for game reasoning with models based on simulation data.%
\footnote{A table of acronyms is provided in Appendix~\ref{app:acronyms}.}
As many multiagent systems of interest are not easily expressed or tackled analytically, EGTA offers an alternative approach whereby a space of strategies is examined through simulation, combined with game model induction and inference.
The number of strategies that can be explicitly incorporated in game models is significantly limited by computational constraints, hence the selection of strategies to include is pivotally important.
For accurate analysis results, we require that the included strategies are high-performing and cover the key strategic issues \citep{balduzzi2019open}.
The challenge of efficiently assembling an effective portfolio of strategies for EGTA is called the \newterm{strategy exploration} problem \citep{Jordan10sw}.


Strategy exploration in EGTA is most clearly formulated within an iterative procedure, whereby generation of new strategies is interleaved with game model estimation and analysis.
The \newterm{Policy Space Response Oracle} (PSRO) algorithm of \citet{Lanctot17} provides a flexible framework for iterative EGTA, where at each iteration, new strategies are generated through reinforcement learning (RL)\@.
The learning player trains in an environment where other players are fixed in a profile (pure or mixed) comprising strategies from previous iterations. 
The key design question is how to set the other-player profile to be employed as a training target. 
In PSRO, the component that derives this target is called a \newterm{meta-strategy solver} (MSS), as it takes an empirical game model as input and ``solves'' it to produce the target profile.
The learning agent then employs RL to search for a strategy best-responding to the MSS target.
In effect, specifying an MSS defines the strategy exploration method for PSRO.


An obvious choice for MSS is the solution concept employed as the objective game analysis, typically Nash equilibrium (NE)\@.
Incrementally adding strategies that are best-responses to NE of the current strategy set is known as the \newterm{double oracle} (DO) algorithm \citep{mcmahan2003planning}, and PSRO with NE as MSS is essentially DO with RL for computing (approximate) best response.
Though DO is often effective, there is ample evidence that best-response to NE is not always the best approach to strategy exploration.
\citet{Schvartzman09a} observed cases where it would approach a true equilibrium extremely slowly, such that even adding random strategies could provide substantial speedups.
More generally, \citet{Lanctot17} argued that best-responding to Nash overfits to the current equilibrium strategies, and thus tends to produce results that do not generalize to the overall space. 
This was indeed their major motivation for defining a generalized MSS concept for strategy exploration.
For example, as an alternative MSS \citet{Lanctot17} proposed \emph{projected replicator dynamics} (PRD), which employs a replicator dynamics (RD) search for equilibrium, truncating the replicator updates to ensure a lower bound on probability of playing each pure strategy.

We take a further step in this direction and adopt an explicit regularization perspective to the specification and analysis of MSSs. 
We propose a novel MSS called \newterm{regularized replicator dynamics} (RRD), which truncates the NE search process in intermediate game models based on a regret criterion. 
Specifically, at each iteration of PSRO, the best-response target profile is updated by running RD, stopping if the regret of the current profile with respect to the empirical game meets a specified regret threshold.
The regret threshold is a hyperparameter, which may be adjusted to suit a particular game class, or annealed to control the degree of regularization across iterations. 
We assess the performance of RRD in various games and show that RRD outperforms several existing MSSs in terms of convergence rate and quality of intermediate empirical game models.

As the size of a payoff matrix is exponential in the number of players, the cost of maintaining completely specified models over the iterations of PSRO can be prohibitive beyond two players.
To mitigate this issue, we employ a PSRO-compatible profile search method, called \newterm{backward profile search} (BPS), which finds solution concepts without simulating the whole payoff matrix.
We combine RRD with BPS,
and demonstrate the effectiveness of this combination in a three-player game.

Finally, our experiments shed light on the source of the benefit of regularization for strategy exploration.
Across a variety of settings, we find that the approximate empirical-game NE produced by RRD tend to have \textit{lower regret in the full game}, compared to exact NE of the empirical game.
This not only provides an explanation for the benefits of regulation, it may also suggest a way to evaluate the potential of novel MSS designs in PSRO-related approaches.

Contributions of this study include:
\begin{enumerate}
    \item RRD: a novel MSS that truncates the NE search process in intermediate game models based on a regret criterion. Our MSS exhibits desired proprieties for strategy exploration (e.g., improved adaptability across games) compared to previous MSSs. 
    We demonstrate that RRD outperforms relevant alternatives from the literature in various games.
    \item A comprehensive analysis of learning with RRD, including the performance stability with respect to the regret threshold.
    \item Integration of method for selective profile evaluation with PSRO, enabling game-solving without simulating the whole payoff matrix. 
    We combine RRD with this method and show its effectiveness for learning in a three-player game.
    \item Demonstration of the key relationship between regret of best-response targets and the performance of MSSs in PSRO.
\end{enumerate}

\section{Related Work on Strategy Exploration}
\label{sec:related}

The first instance of automated strategy generation in EGTA was a genetic search over a parametric strategy space, optimizing performance against an equilibrium of the empirical game \citep{Phelps06}. 
\citet{Schvartzman09} deployed tabular RL as a best-response oracle in EGTA for strategy generation.
These same authors framed the general problem of \newterm{strategy exploration} in EGTA and investigated whether better options exist beyond best-responding to an equilibrium \citep{Schvartzman09a}.
\citet{Jordan10sw} further extended this line of work by adding strategies that maximize the deviation gain from an empirical rational closure.

Investigation of strategy exploration was advanced significantly by introduction of the PSRO framework \citep{Lanctot17}.
PSRO applied deep RL as an approximate best-response oracle to certain designated other-agent profile selected by the MSS.
When employing NE as MSS, PSRO reduces to the DO algorithm \citep{mcmahan2003planning}.
To generate strategy effectively, \citet{Lanctot17} balanced between overfitting to NE and generalizing to the strategy space outside the empirical game, and proposed \newterm{projected replicator dynamics} (PRD), which employs an RD search for equilibrium \citep{taylor1978evolutionary,smith1973logic} and ensures a lower bound on probability of playing each pure strategy. 
For simplicity, we often refer to an MSS as shorthand for PSRO with that MSS when the context is unambiguous.
For example, we may say ``PRD'' to mean ``PSRO with PRD''.

PSRO can also be viewed as generalizing some classic game-learning dynamics.
For example, selecting a uniform distribution over current strategies as MSS essentially reproduces the classic \newterm{fictitious play} (FP) algorithm \citep{brown1951iterative}. 
Moreover, an MSS that simply extracts the most recent strategy duplicates the \newterm{iterated best response} algorithm. 
Note that the MSSs generating these dynamics are solvers in only a trivial sense; they do not substantively employ an empirical game model as they derive from the strategy sets directly.

Following the line of PSRO, some works propose MSSs that effectively \textit{regularize} the target profile to prevent from best-responding to an exact equilibrium. 
Specifically, \citet{wang19sywsjf} employed a mixture of NE and uniform, which essentially samples whether to apply DO or FP for every PSRO iteration, thus illustrating the possibility of combining MSSs.
\citet{Wright19} added an adjustment step to DO, which fine-tunes the generated policy network against a mix of previous equilibrium strategies. 
\citet{balduzzi2019open} introduced a new MSS, called \newterm{rectified Nash}, designed to increase diversity of empirical strategy space. 
\citet{Dinh22} proposed an MSS for two-player zero-sum game that applies online learning to the empirical game and outputs the online profile as a best-response target. 
Beyond selecting NE as a solution concept, \citet{muller2019generalized} proposed a new MSS based on an evolutionary-based concept, $\alpha$-rank \citep{omidshafiei2019alpha},  and \citet{marris2021multi} proposed maximum welfare coarse correlated equilibrium (MWCCE), and maximum Gini coarse correlated equilibrium (MGCCE) for computing correlated equilibria, within the PSRO framework.

\section{Preliminaries}
A normal-form game $\mathcal{G}=(N,(S_i),(u_i))$ comprises a finite
set of players $N$ indexed by $i$, a non-empty set of strategies
$S_i$ for player $i\in N$, and a utility function $u_i: \prod_{j \in N}S_j \rightarrow \mathbb{R}$ for player $i\in N$.

A mixed strategy $\sigma_i$ is a probability distribution over strategies in $S_i$, with $\sigma_i(s_i)$ denoting the probability player~$i$ plays strategy $s_i$. 
We adopt conventional notation for the other-agent profile: $\sigma_{-i}=\prod_{j \ne i}\sigma_j$. 
Let $\Delta(\cdot)$ represent the probability simplex over a set. The mixed strategy space for player~$i$ is given by $\Delta(S_i)$. Similarly, $\Delta(S)=\prod_{i \in N}\Delta(S_i)$ is the mixed profile space.

Player~$i$'s \newterm{best response} to profile $\sigma$ comprises strategies yielding maximum payoff for~$i$, fixing other-player strategies:
\begin{displaymath}
\br_i(\sigma_{-i}) \equiv \argmax_{\sigma_i'\in \Delta(S_i)} u_i(\sigma_i', \sigma_{-i}).
\end{displaymath}
Let $\br(\sigma)\equiv\prod_{i \in N}\br_i(\sigma_{-i})$ be the overall best-response correspondence for a profile $\sigma$. A \newterm{Nash equilibrium} (NE) is a profile $\sigma^*$ such that $\sigma^*\in \br(\sigma^*)$. 

Player~$i$'s \newterm{regret} for profile $\sigma$ in game $\mathcal{G}$ is given by 
\begin{displaymath}
    \rho^{\mathcal{G}}_i(\sigma) \equiv \max_{s_i'\in S_i}u_i(s_i', \sigma_{-i})-u_i(\sigma_i, \sigma_{-i}).
\end{displaymath}
Regret captures the maximum player~$i$ can gain in expectation by unilaterally deviating from its mixed strategy in $\sigma$ to an alternative strategy in $S_i$\@. 
An NE has zero regret for every player. 
A profile is said to be an \newterm{$\epsilon$-Nash equilibrium} ($\epsilon$-NE) if no player can gain more than $\epsilon$ by unilateral deviation. 
We define the regret of a strategy profile $\sigma$ as the sum over player regrets:
\begin{displaymath}
    \rho^{\mathcal{G}}(\sigma) \equiv \sum_{i\in N} \rho^{\mathcal{G}}_i(\sigma).
\end{displaymath}

A \newterm{restricted game} $\mathcal{G}_{S\downarrow X}$ is a projection of full game~$\mathcal{G}$, in which players choose from restricted strategy sets $X_i\subseteq S_i$.
An \newterm{empirical game} $\mathcal{\hat{G}}$ is a model of true game $\mathcal{G}$ where payoffs are estimated through simulation.
Thus, $\mathcal{\hat{G}}_{S\downarrow X}=(N,(X_i),(\hat{u}_i))$ denotes an empirical game model where $\hat{u}$ is an estimated projection of $u$ onto the strategy space~$X$\@.

\newterm{Replicator dynamics} (RD) describes an evolving trajectory of mixed profiles, inspired by natural selection \citep{taylor1978evolutionary,smith1973logic}.
RD is commonly employed as a heuristic equilibrium search algorithm. 
We consider a discrete form of RD, where player~$i$'s probability of playing strategy is updated in proportion to its payoff for deviating to that strategy from the current mixture. 
Mathematically, the replicator equation for player $i$'s strategy $s_i$ in a current profile $\sigma$ is given by
\begin{displaymath}
     \frac{d\sigma_i(s_i)}{dt} = \sigma_i(s_i)[u_i(s_i,\sigma_{-i}) - u_i(\sigma_i, \sigma_{-i})].
\end{displaymath}
At each iteration of RD, player $i$'s mixed strategy $\sigma_i$ is updated by $\sigma_i\gets P(\sigma_i + \alpha \frac{d\sigma_i}{dt})$, where $\alpha$ is a step size for RD and $P$ is a projection operator to the strategy simplex, namely $P(\sigma_i)=\argmin_{\sigma_i'\in \Delta}\norm{\sigma_i'-\sigma_i}_2$. 


\newterm{PSRO} is presented below as Algorithm~\ref{alg: PSRO}.
Choice of MSS dictates the strategy exploration approach.

\begin{algorithm}[H]
    \caption{PSRO, parametrized by solver MSS}
    \label{alg: PSRO}
    \begin{algorithmic}[1]
        \REQUIRE initial strategy sets $X$
        \STATE Estimate $\mathcal{\hat{G}}_{S\downarrow X}$ by simulating  $\sigma \in X$
        \STATE Initialize target $\sigma_i \gets \text{Uniform}(X_i)$ 
        \FOR{PSRO iteration $\tau=1,2,\dotsc, \mathcal{T}$}
        \FOR{player $i\in N$}
        \FOR{many RL training episodes}
        \STATE Sample a profile $s_{-i}\in \sigma_{-i}$
        \STATE Train BR oracle $s_i'$ against $s_{-i}$
        \ENDFOR
        \STATE $X_i \gets X_i \cup \{s_i'\}$
        \ENDFOR
        \STATE Update $\mathcal{\hat{G}}_{S\downarrow X}$ by simulating missing profiles over $X$
        \STATE Compute best-response target $\sigma \gets \text{MSS}(\mathcal{\hat{G}}_{S\downarrow X})$ 
        \ENDFOR
        \STATE \textbf{Return} $\mathcal{\hat{G}}_{S\downarrow X}$
    \end{algorithmic}
\end{algorithm}

\section{Regularization for Strategy Exploration}
\subsection{Regularized Replicator Dynamics}

To avoid overfitting a response to NE, we adopt an explicit regularization perspective on strategy exploration.
Specifically, we propose a method to derive approximate NE by truncating an RD-based search.
Our new MSS, called \newterm{regularized RD} (RRD), simply runs RD on the empirical game, stopping when the regret of the current profile (w.r.t the empirical game) meets a specified regret threshold $\lambda$, or a maximum number of iterations is reached.
In the RRD procedure (Algorithm~\ref{alg: RD}), each player's strategy is initialized with a uniform distribution over strategies in the empirical game. 
Then the replicator equation is iteratively applied until the regret of the current profile (w.r.t the empirical game) becomes smaller than the regret threshold $\lambda$. 
Since RD does not generally converge to an exact equilibrium, there is no guarantee a finite regret threshold $\lambda$ will ever be reached.
We therefore set a maximum number of iterations $M$, and if the limit is reached return the profile with the lowest regret found to that point.

Note that RRD supports direct control of the degree of regularization through an explicit parameter: the regret threshold.
This parameter is meaningful across games with different strategy sets, as long as the utility scales on which regret is measured are comparable.

\begin{algorithm}[H]
    \caption{RRD}\label{alg: RD}
    \textbf{Parameters}: regret threshold $\lambda$, RD step size $\alpha$ \\
    \textbf{Input}: an empirical game $\mathcal{\hat{G}}_{S\downarrow X}$ 
    \begin{algorithmic}[1]
        \STATE Initialize RD with $\sigma_i \gets \text{Uniform}(X_i)$
        \WHILE{$\rho^{\mathcal{\hat{G}}_{S\downarrow X}}(\sigma) > \lambda$}
        \FOR{player $i\in N$}
        \STATE $\sigma_i\leftarrow P(\sigma_i + \alpha \frac{d\sigma_i}{dt})$
        \ENDFOR
        \ENDWHILE
        \STATE \textbf{Return} $\sigma$
    \end{algorithmic}
\end{algorithm}


The procedure of PSRO with RRD is obtained by employing RRD as the MSS in Algorithm~\ref{alg: PSRO}.
A slight modification readily supports annealing schemes where the regret threshold for RRD is varied across PSRO iterations.

\subsection{Selective Profile Evaluation using BPS}
\label{sec:multiple players}
One obstacle to scaling PSRO is that the size of the empirical game grows exponentially in the number of players. 
Even in games with only a few players (e.g., three or four), exhaustive simulation of the payoff matrix may become infeasible as the strategy space grows.
Fortunately, it is often possible to derive solutions from only partial payoff matrices \citep{Fearnley15}, and prior EGTA researchers have developed methods for selectively evaluating strategy profiles in search for equilibrium \cite{Jordan08vw,Sureka:2005uq}.
For example, the algorithm of \citet{brinkman2016shading} maintains a set of \newterm{complete subgames} of the empirical game (i.e., maximal strategy sets over which profiles have been exhaustively evaluated). 
NE of these subgames are proposed as candidate equilibria of the empirical game. 
For each candidate, all of the one-player deviations to strategies outside the subgame are evaluated.
If there is no beneficial deviation to the candidate, then the candidate is \newterm{confirmed} as an NE of the empirical game. 
Otherwise, subgames are extended by the beneficially deviating strategies and the search continues.

We developed a simple profile search method for PSRO, which we call \newterm{backward profile search} (BPS, Algorithm~\ref{alg:Simple Backward Profile Search}).
Our method resembles that of \citet{brinkman2016shading}, but takes into account the sequence in which the strategies were generated.
At each iteration, BPS starts search from the strategies most recently added to the empirical game by PSRO, then searches potential deviations backward across previous PSRO iterations. 
Our motivation is that the newest strategies are most likely to participate in equilibria.
Once BPS confirms a solution of the empirical game, we apply RRD to the subgame over the support of this solution. 
By construction, this subgame is completely evaluated (as required by RD), whereas the entire empirical game payoff matrix is only partially evaluated.
In our experiments, we show that BPS can successfully find best-response targets in a three-player game, short of exhaustive evaluation of the empirical game.

\begin{algorithm}[!hptb] 
\caption{Backward Profile Search}
\label{alg:Simple Backward Profile Search}
\begin{algorithmic}[1] 
\REQUIRE Empirical game $\mathcal{\hat{G}}_{S\downarrow X}$ with partial payoff matrix.\\
\STATE Initialize subgame with strategy sets $Z=(Z_i)$, $i\in N$, where $ Z_i=\{s_\tau^i\}$, with $s_\tau^i$ the player~$i$ strategy added in the most recent PSRO iteration $\tau$.
\WHILE{True}
\STATE $\sigma \gets \text{NE\_search}(\hat{G}_{S\downarrow Z})$
\STATE deviation\_exists $\gets$ False
\FOR{player $i\in N$}
\STATE $s_i \gets \argmax_{s' \in X_i} u_i(s', \sigma_{-i}) - u_i(\sigma_i, \sigma_{-i})$
\IF{$s_i \notin Z_i$}
\STATE $Z_i \gets Z_i \cup \{s_i\}$
\STATE deviation\_exists $\gets$ True
\ENDIF
\ENDFOR
\IF{$\neg$ deviation\_exist}
\STATE \textbf{return} $\sigma$
\ENDIF

\STATE Evaluate missing profiles of $Z$ through simulation.
\ENDWHILE
\end{algorithmic}
\end{algorithm}

Figure~\ref{fig:cubes} illustrates the mechanism of BPS at the third iteration of PSRO, in which each player has four strategies. 
The $4\times 4\times 4$ cube in Figure~\ref{fig:cubes} represents the payoff matrix of the current empirical game. 
Green cells represent the payoffs of profiles that have been evaluated from previous iterations and white cells represent the payoffs of potential deviations from the equilibrium of the current subgame. 
The missing cells represent payoffs of profiles that have not been evaluated.
Note that the current payoff matrix is incomplete. 
To find the NE of the empirical game, BPS starts search from evaluating 
the subgame constituted by the most recently added strategy of each player, represented by the red cell.
Since the red cell is a pure-strategy profile, it is also the NE of the current subgame.
Next, BPS evaluates the payoffs of all potential deviations (white cells) from the red cell.
Suppose the blue cell is a profile with the largest deviation payoff for certain player $i$.
Then player $i$ adds the corresponding deviation strategy to her strategy set of the current subgame.
Now the profile space of the current subgame contains the red and the blue profiles.
Then BPS repeats evaluating all profiles in the current subgame, computing NE of the current subgame, and evaluating potential deviations from the NE.
Again, if the purple cell is a deviation profile, then the corresponding deviating strategy will be added to the strategy set of the subgame.
BPS repeats this procedure until the NE of the subgame is confirmed, that is, no beneficial deviation could be found in the empirical game.

\begin{figure}[!htpb]
\centering
\includegraphics[width=0.8\columnwidth]{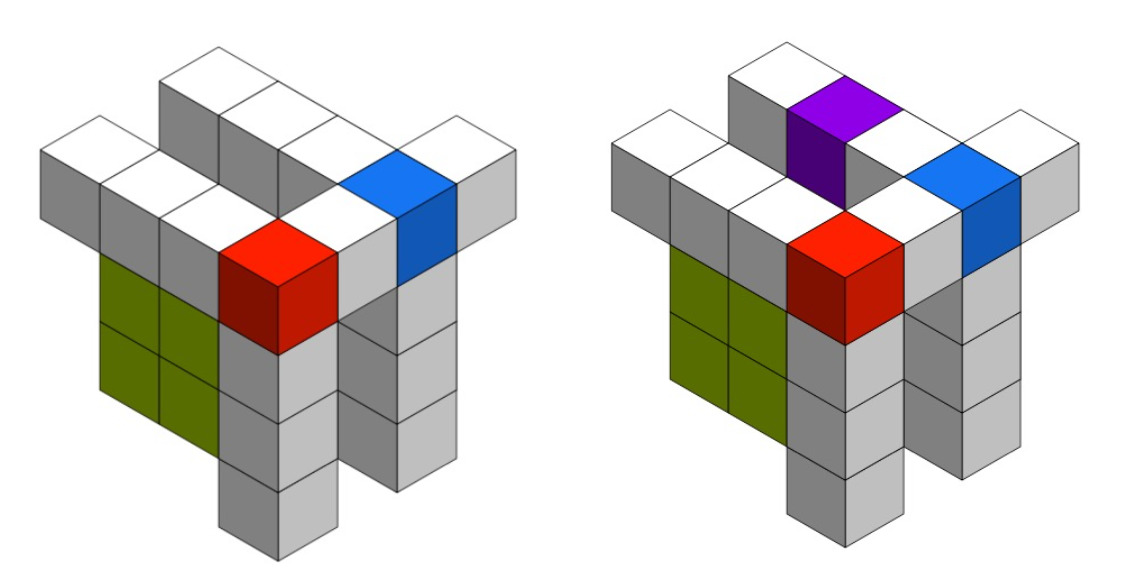}
\caption{An illustration of a partial payoff matrix of a three-player empirical game and the workflow of BPS. 
Green: evaluated profiles from previous PSRO iterations; White: deviation profiles from NE of current subgame; Red: the profile with the most recently added strategies; Blue and purple: profiles with the largest deviation payoffs.}\label{fig:cubes}
\end{figure}

\section{Performance and Analysis}
\subsection{Experimental Results}\label{sec: results}

\subsubsection{Two-player Leduc Poker}
In Figure~\ref{fig:regularization Leduc}, we test our algorithm on two-player Leduc poker and plot the regret curves (w.r.t the full game) given by FP, DO, PSRO with PRD, and PSRO with RRD (under two stopping criteria).
We first observe that RRD yields a rapid convergence to a low-regret value compared to other MSSs.
It is quite striking that RRD outperforms PRD (prior best known for this game) by such a large margin.

To show the benefits of using a regret threshold as a stopping criterion compared to a fixed number of RD updates, we plot the best regret curve of RRD using a fixed number of RD updates.
We observe that RRD performs better using a regret threshold. 
This is because the number of RD updates that produces the right level of regularization varies across empirical games. 

\begin{figure}[!htpb]
\centering
\includegraphics[width=0.95\columnwidth]{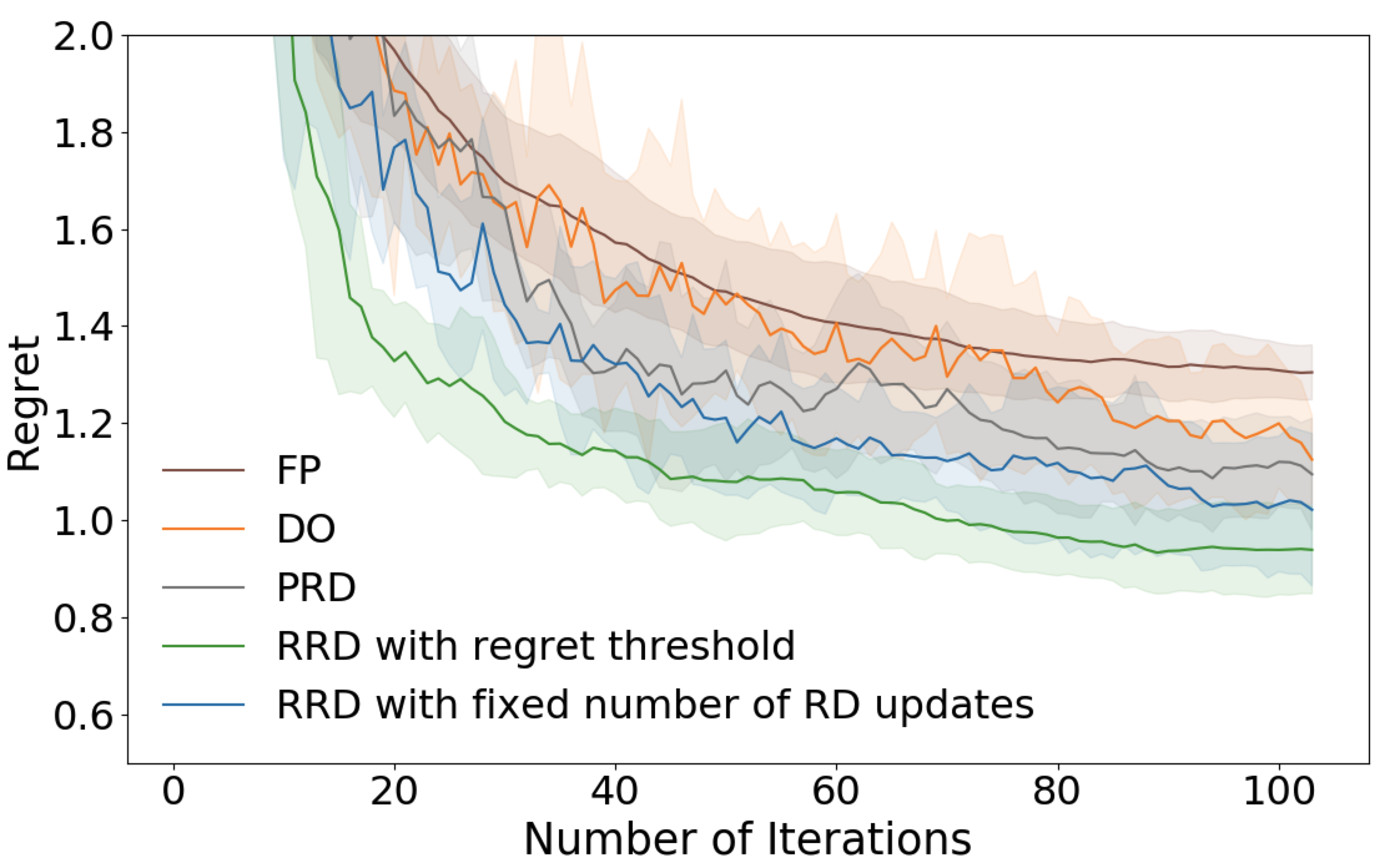}
\caption{RRD performance in two-player Leduc Poker.}\label{fig:regularization Leduc}
\end{figure}


\subsubsection{Real-World Games}
We further evaluate our algorithms in four of the ``real-world games'' studied by \citet{czarnecki2020real}: Hex, Connect four, Misere Tic Tac Toe, and Go. 
We observe that RRD exhibits faster convergence than FP and DO in all four games. 
We report further experimental details and full game descriptions in Appendix~\ref{app:std of hex}.


\subsubsection{Multi-Player Games}
We apply the combination of BPS and RRD to three-player Leduc poker. 
As shown in Figure~\ref{fig: RD_crd_3p}, although RRD is applied only to the subgame containing the support of exact equilibrium, learning  still benefits from regularization. 
In Table~\ref{tab: sims}, we list the average of number of profiles evaluated in the empirical game at different PSRO iterations with and without BPS.
To evaluate a profile, we estimate payoffs by averaging over 1000 samples. 
Employing BPS in this example saved approximately 10\% simulation effort, compared to exhaustive estimation. 
We can also see that evaluation savings become somewhat more significant as the number of iteration increases.
For example, from iterations 41 to 50 we evaluated 53890 profiles, while the profile space increased by 61000, a savings of 11.7\%. 
Experience in prior EGTA work demonstrates that the fractional savings are substantially greater for games with more than three players \citep{brinkman2016shading}.


\begin{figure}[!htpb]
\centering
\includegraphics[width=0.93\columnwidth]{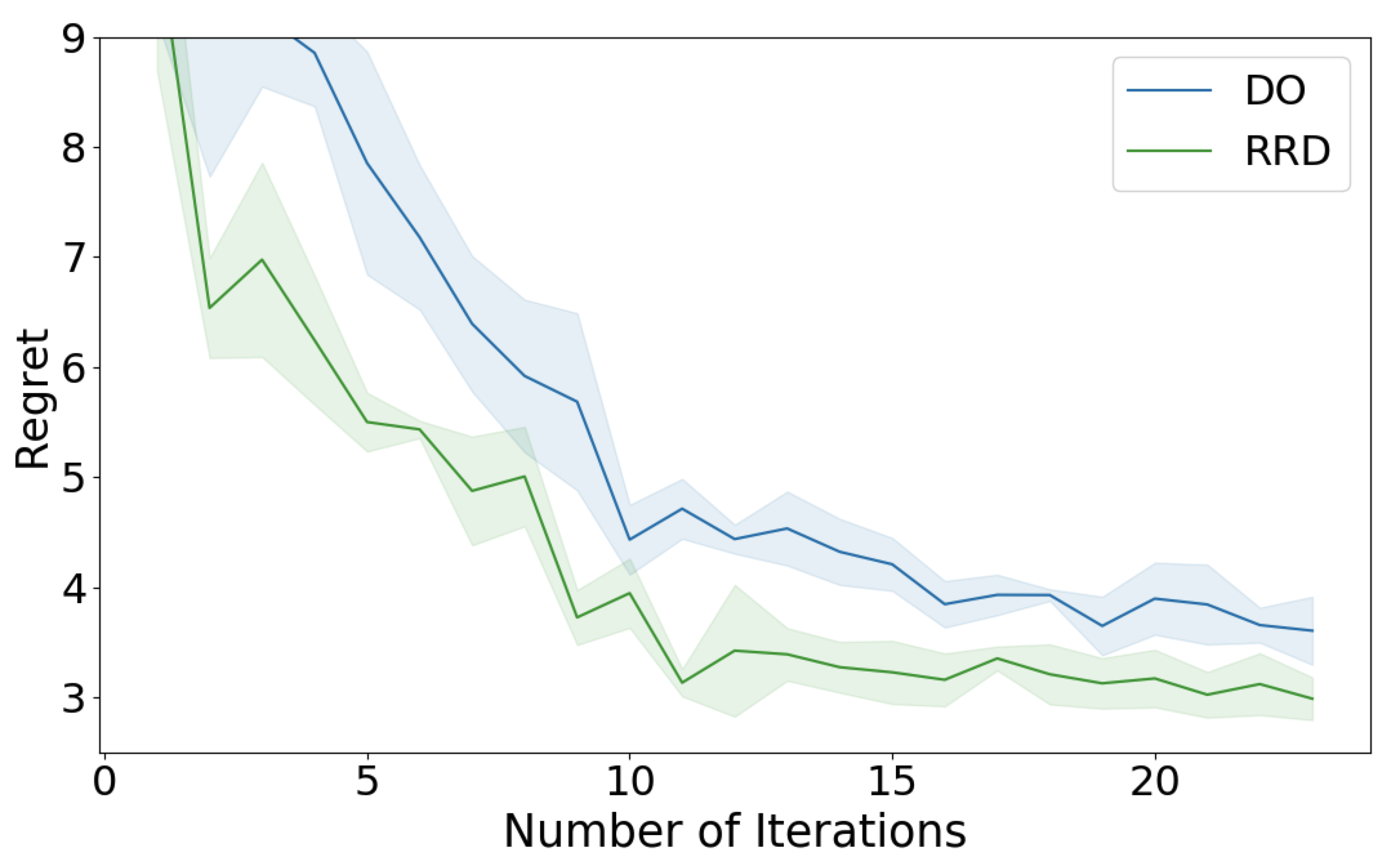}
\captionof{figure}{RRD performance in three-player Leduc poker.}
\label{fig: RD_crd_3p}
\end{figure}

\begin{table}[!ht]
\centering
\begin{tabular}{ rrrr }
    \toprule
    Iter\# & \#Profiles w. BPS & \#Profiles & Savings Pct. \\
    \midrule
    10 & 880   & 1000 & 12.0\%\\
    20 & 7100  & 8000 & 11.1\%\\
    30 & 24667 & 27000 & 7.5\%\\
    40 & 58400 & 64000 & 8.8\%\\
    50 & 112290& 125000 & 11.7\%\\
    \bottomrule
\end{tabular}
\caption{Savings in profile simulation due to BPS in a three-player game.}
\label{tab: sims}
\end{table}



\begin{figure*}[!t]
\centering
    \subfloat[Hex.]{\includegraphics[ width=0.245\textwidth]{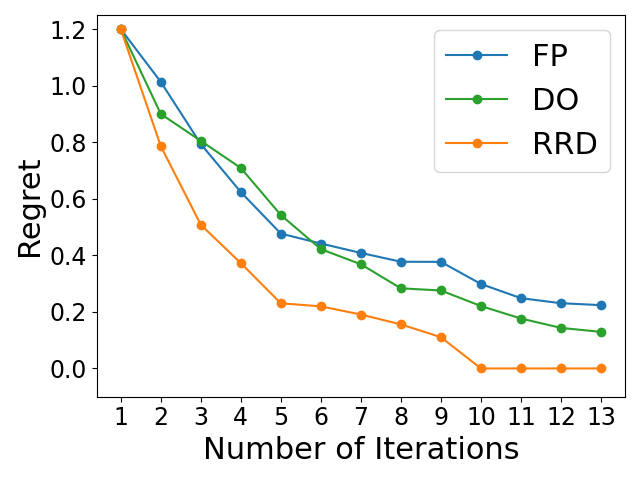}\label{fig:hex}}\
     \subfloat[Connect four.]{\includegraphics[width=0.245\textwidth]{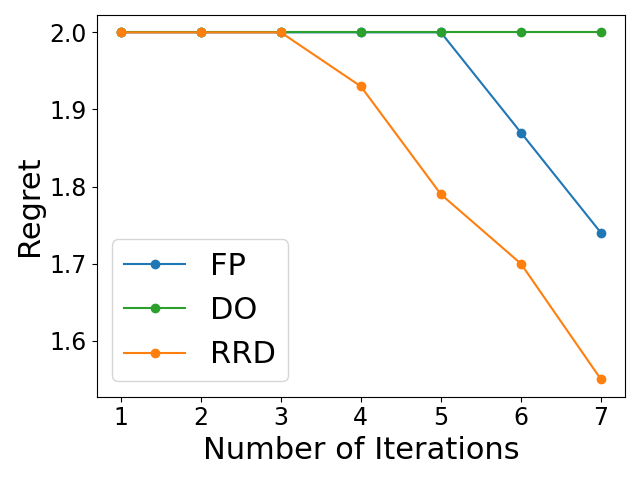}\label{fig:connect_four}}\
    \subfloat[Misere Tic Tac Toe.]{\includegraphics[ width=0.245\textwidth]{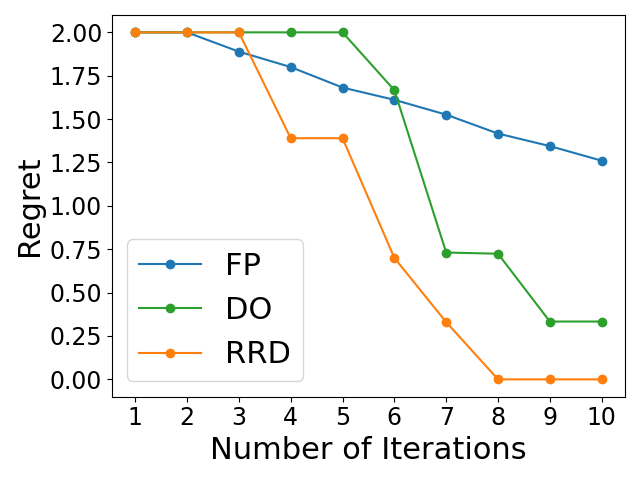}\label{fig:misere_tic}}\
    \subfloat[Go (size=4).]{\includegraphics[width=0.245\textwidth]{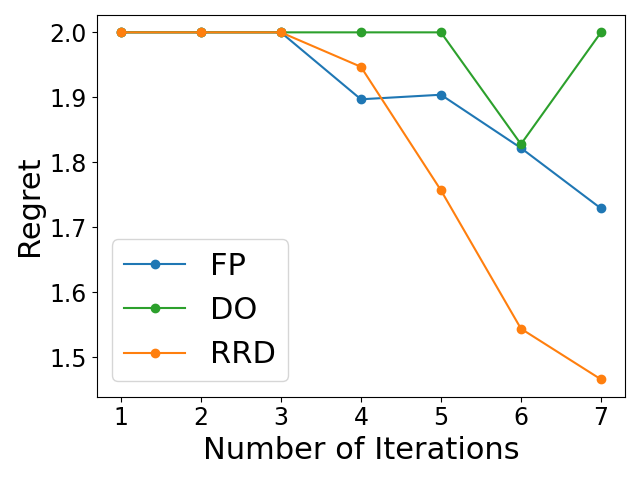}\label{fig:go}}\\
\caption{RRD performance compared to FP and DO in four games studied by \protect\citet{czarnecki2020real}.}
\end{figure*}

\subsubsection{Attack-Graph Games}
\newterm{Attack graphs} \cite{miehling2015planning} are a tool in cyber-security analysis employed to model the paths by which an adversary may compromise a system. 
An \emph{attack-graph game} is a two-player general-sum game defined on the attack graph where an attacker attempts to compromise a sequence of nodes to reach \emph{goal} nodes and a defender endeavors to protect any node (e.g., deny an access).
Reaching the \emph{goal} nodes within a finite horizon provides a large benefit for the attacker and a substantial loss for the defender.
Both offensive and defensive actions are associated with a cost.
The ability of the attacker to choose any subset of feasible nodes and of the defender to defend any subset of the nodes induces action spaces of combinatorial size. 

In Figure~\ref{fig: att_graph}, we show the performance of RRD on an attack-graph game instance with 100 nodes and hence $2^{100}$ possible combinatorial actions. 
Since the game is too large to analyze exhaustively, we first construct a particular set of deep Q-network (DQN) \cite{mnih2013playing} strategies with 125 strategically-diverse strategies in total, following the strategy sampling approach by \citet{czarnecki2020real}. Then we apply game-theoretic analysis (i.e., FP, DO, and RRD) to this set of strategies.
Each regret curve is an average over 5 randomly-selected initial strategies.
From Figure~\ref{fig: att_graph}, we observe that even though the game of interest is large and beyond two-player zero-sum games, RRD still boosts faster convergence and less variance than DO and FP\@.

\subsubsection{Bargaining Games}

We consider a non-zero-sum incomplete-information bargaining game, in which two players engage in a sequential process to reach a deal over a vector of items \citep{lewis2017deal}. 
In our setting, there are three items, and up to $T=10$ rounds of offers. 
The value of each item is sampled from a commonly known distribution for both players and each player only knows their own value realizations.
If a deal is reached at time $t\le T$, players receive their corresponding values for the items specified, discounted by a factor of $\gamma^t$ ($\gamma = 0.9$). 
If no deal is reached, both players receive zero utility. 
We represent negotiation policies using a neural network, and employ DQN to compute approximate best responses.

In Figure~\ref{fig: bargaining}, we consider social welfare (SW) as a performance measure and compare the averaged social welfare of the same solution concept in the empirical games given by different MSSs (i.e., RRD, DO, FP, MWCCE, and MGCCE). 
Each SW value is an average over 5 runs and we show the standard deviations in the Appendix~\ref{app:SD bargain}.
We select five solution concepts for performance comparison, including maximum SW pure strategy profile (Max SW), NE, uniform distribution over strategies, MWCCE, and MGCCE. 
From Figure~\ref{fig: bargaining}, we observe that RRD can generate all five solutions with higher social welfare than others. 
This observation also indicates that the best MSS to approximate a solution (e.g., NE, MWCCE, and MGCCE) may not be the one that uses the solution concept directly as a best response target.

\subsection{Stability with Varying Regret Threshold}

To investigate the stability of learning performance w.r.t the regret threshold $\lambda$, we select a wide range of $\lambda$s for RRD and compare the regrets at the last iteration of PSRO under these $\lambda$s with the regret of DO in two-player Leduc poker. 
We plot the regrets in Figure~\ref{fig:Regularization range}. 
From Figure~\ref{fig:Regularization range}, we observe that all $\lambda$s in the range yield a better learning performance than DO, which demonstrates the stability of the performance of RRD w.r.t the regret threshold $\lambda$.
In addition, we observe that as the value of regret threshold $\lambda$ increases, the learning performance first improves and then becomes worse. 
This means that either excessive or inadequate regularization would damage the overall learning performance.

\begin{figure}[!htpb]
\centering
\includegraphics[width=0.8\columnwidth]{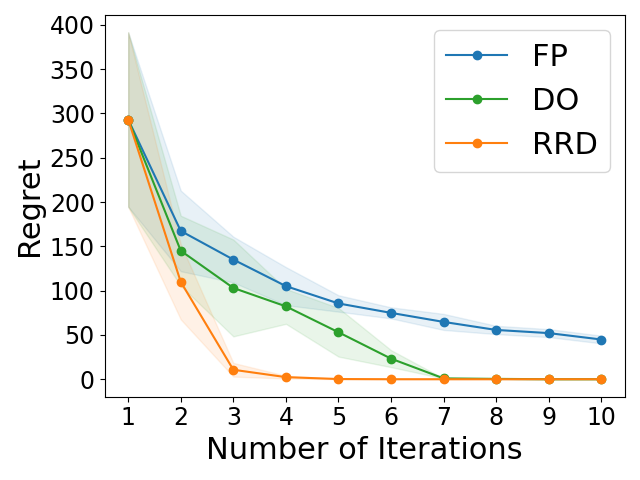}
\captionof{figure}{RRD outperforms FP and DO in the attack-graph game.}
\label{fig: att_graph}
\end{figure}



\begin{figure*}

 \center

  \includegraphics[width=\textwidth]{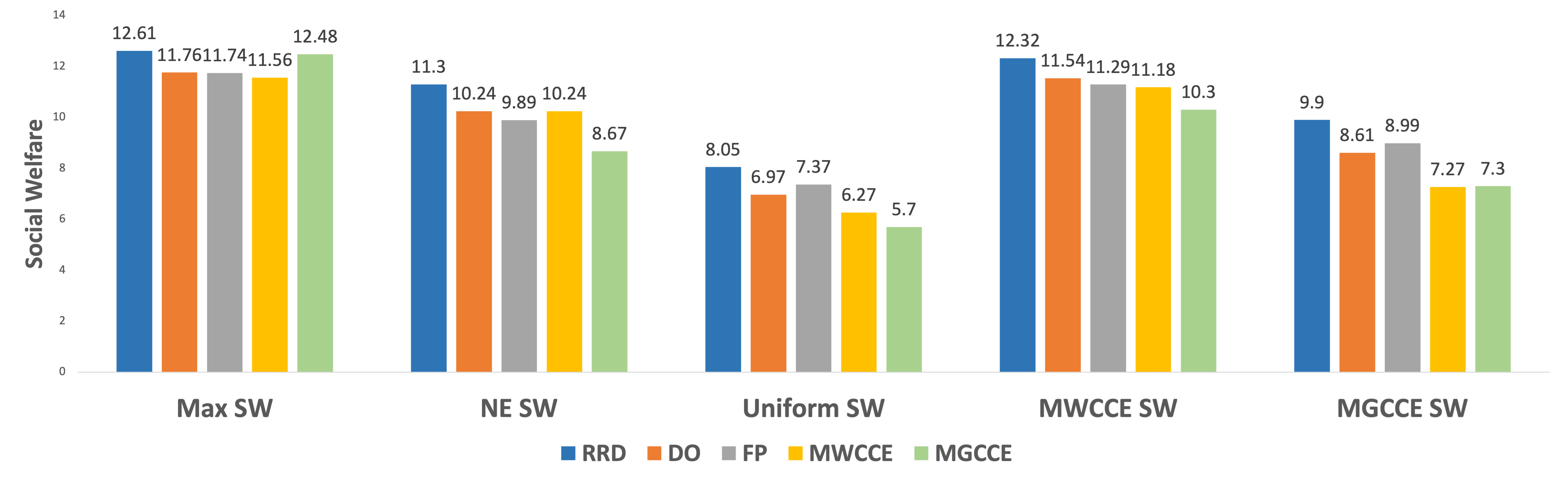}

  \caption{RRD performance in bargaining games. Each color represents an MSS and each bundle of colors shows the SW of a given solution concept in the corresponding empirical games.}

  \label{fig: bargaining}

\end{figure*}

\begin{figure*}[!ht]
\centering
    \subfloat[Range of regret thresholds.]{\includegraphics[width=0.495\textwidth]{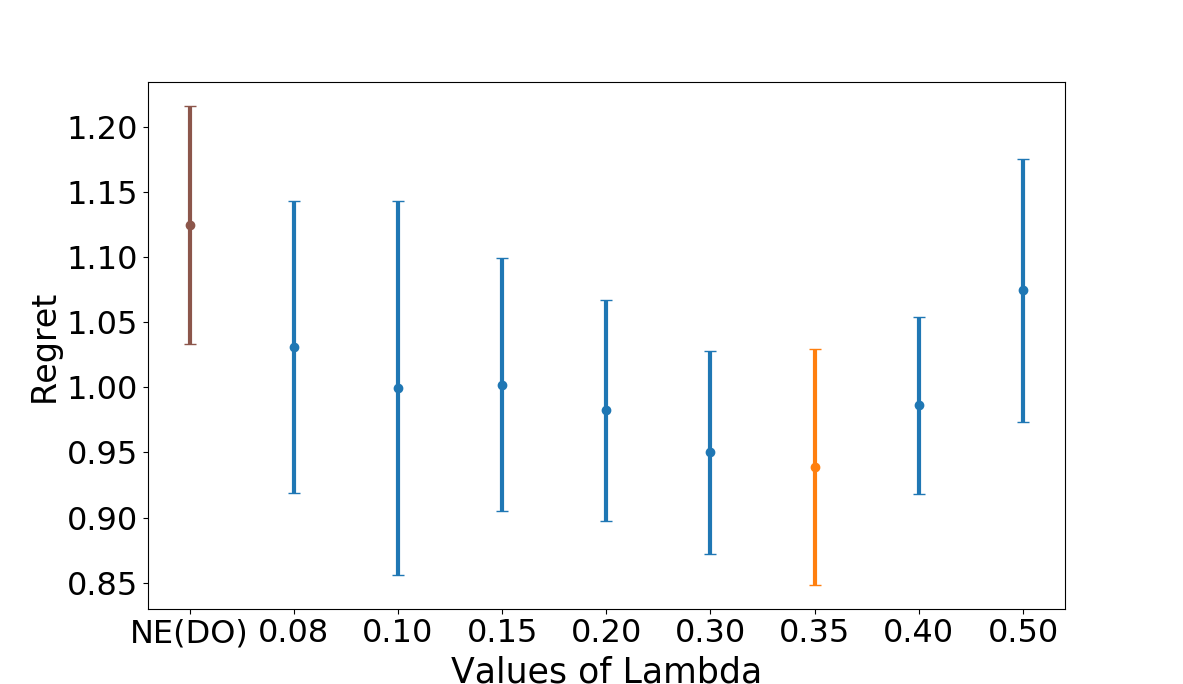}\label{fig:Regularization range}}\
    \subfloat[Decreased regret when regularization is applied.]{\includegraphics[width=0.435\textwidth]{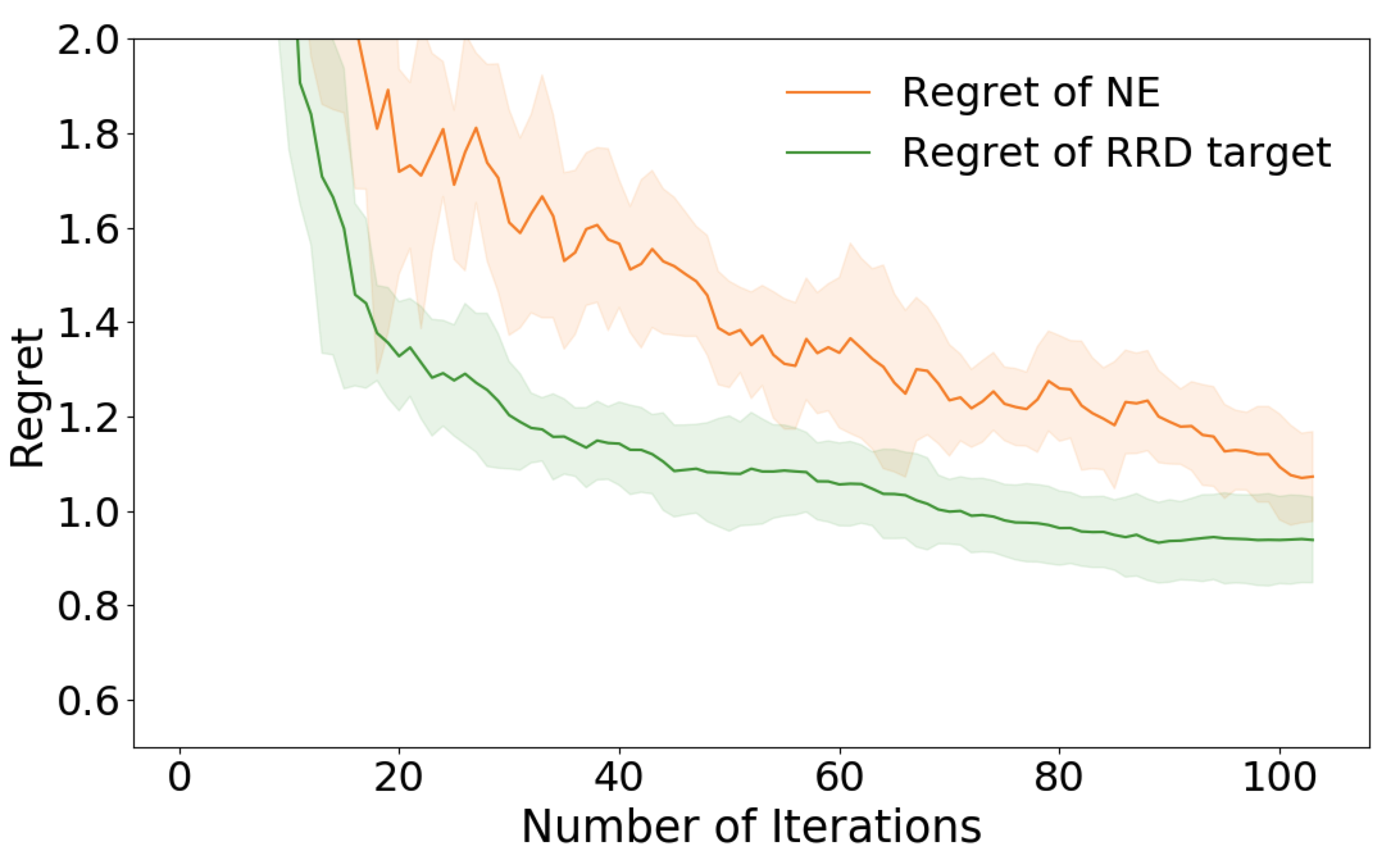}\label{fig: RRD NEvsRRD}}\\
\caption{Properties of learning with RRD in two-player Leduc Poker.}
\end{figure*}


\section{Explaining the Performance of MSSs}

Prior studies of strategy exploration recognized that best-responding to exact NE is not ideal, and demonstrated improvements over DO through alternative MSSs or other approaches.
However, to date there has been no satisfactory explanation for what makes for an effective best-response target.
We briefly discuss some prior studies, then provide a novel insight---based on our experimental observations---which helps to explain why regularization is helpful in this context.

\subsection{Prior Regularization-Related Approaches}

In an early study of strategy exploration, \citet{Schvartzman09a} found that adding noise to NE alters the path of equilibrium search and accelerated the overall learning. 
But how this noise contributed to the acceleration was unexplained. 
As noted above, the originators of PSRO introduced PRD expressly to address overfitting to NE \citep{Lanctot17}.
PRD promotes exploration by ensuring each strategy in the empirical game a fixed minimal probability in the best-response target. 
With similar motivation, \citet{wang19sywsjf} proposed to alternate FP and DO randomly, and \citet{Wright19} adjusted the best response to NE by tuning against previous opponents.

Related to this line of work, \citet{balduzzi2019open} introduced the term \newterm{Gamescape} to describe the conceptual strategy space covered by the empirical game.
They proposed an MSS called \newterm{rectified Nash}, designed to qualitatively extend the Gamescape.
However, the concept of Gamescape was described primarily through illustration in a particular simple game, which is instructive but does not provide operational definitions or measures. 

\subsection{A Novel Explanation}

Our key insight is that the performance of strategy exploration is strongly related to the regret of best-response targets \textit{w.r.t the full game}.
To illustrate this phenomenon, Figure~\ref{fig: RRD NEvsRRD} presents regret curves for PSRO with RRD in two-player Leduc poker.
The two curves show true-game regret for NE and RRD, respectively, as computed at each PSRO iteration.
Note that throughout the run, the regret of the RRD solution is much smaller than that of the empirical NE\@. 
We observe the same phenomenon in PSRO runs generated by other MSSs (see Appendix~\ref{app:decreased regret}).
In other words, whereas RRD has higher regret than NE in the empirical game ($\lambda$ versus zero), it reliably has lower regret in the full game.
Since our ultimate objective is a full-game low-regret solution, this helps to explain why the regularization imposed by RRD apparently provides robustly improved performance for strategy exploration. 

Note that this observation only goes so far; it is not the case that minimizing full-game regret always provides the optimal best-response target for strategy exploration.
This is because the lowest full-game regret profile may not change much from one PSRO iteration to the next, and so selecting targets on that basis may compromise the diversity of constructed empirical games. 
We tested this explicitly by using as an MSS the profile in the empirical game that has the lowest regret with respect to the full game.
Our results confirm that the extreme choice of target is indeed suboptimal for strategy exploration (details provided in Appendix~\ref{app: MRCP as MSS}).

\section{Conclusion}

We propose RRD as a novel MSS for PSRO, explicitly based on regularization. 
By controlling the regret threshold, the degree of regularization can be adjusted to suit a particular strategy exploration context.
In our experiments, we show that RRD outperforms several existing MSSs in various games and investigate many properties of learning with RRD\@. 
To help scale beyond two-player games, we propose BPS, a PSRO-compatible profile search method that avoids exhaustive simulation of the game matrix.
We show the benefit of regularization when combining BPS with RRD in three-player Leduc poker. 
Finally, we demonstrate that the performance of strategy exploration is strongly related to the regret of best-response targets and regularization could significantly decrease the regret of best-response targets, thus contributing to an improved learning.


\bibliographystyle{named}
\bibliography{thebib,wellman}


\appendix
\newpage
\clearpage

\section{Table of Acronyms}\label{app:acronyms}

\begin{table}[!htb]
    \centering
    \begin{tabular}{ll}
        \toprule
        \textbf{Abbreviation} & \textbf{Definition}\\
        \midrule
        BPS & Backward Profile Search\\
DO 			&	Double Oracle\\
EGTA 		&		Empirical Game-Theoretic Analysis\\
FP 		&		Fictitious Play\\
MSS 		&		Meta-Strategy Solver\\
MRCP 		&		Minimum Regret Constrained Profile\\
NE 		&		Nash Equilibrium\\
PRD 	&			Projected Replicator Dynamics\\
PSRO 	&			Policy Space Response Oracle\\
QRE & Quantal Response Equilibrium\\
RD & Replicator Dynamics\\
RL 		&		Reinforcement Learning\\
RRD & Regularized Replicator Dynamics\\
        \bottomrule
    \end{tabular}
    \caption{Table of acronyms in alphabetical order.}
\end{table}

\section{Extra Experimental Results}

\subsection{Advantage over RD with a Fixed Number of iterations}\label{app: fixed iter}

 

In Figure~\ref{fig:RRD iter}, we show the number of RD updates needed to reach the regret threshold $\lambda$.
We fix the lower bound of the number of iterations to 10000.
As shown in Figure~\ref{fig:RRD iter}, the number of RD updates required to reach the regret threshold $\lambda$ varies across PSRO iterations, which again emphasizes the need of adjusting the number of RD updates dynamically rather than fixing the number of RD updates.
Moreover, it is interesting to observe that the number of iterations is higher in the middle of the learning while it is lower at both the beginning and the end.
This contradicts the stereotype that RD needs more iterations to converge in an empirical games with more diverse strategies. 

\begin{figure}[!htpb]
\centering
\includegraphics[width=\columnwidth]{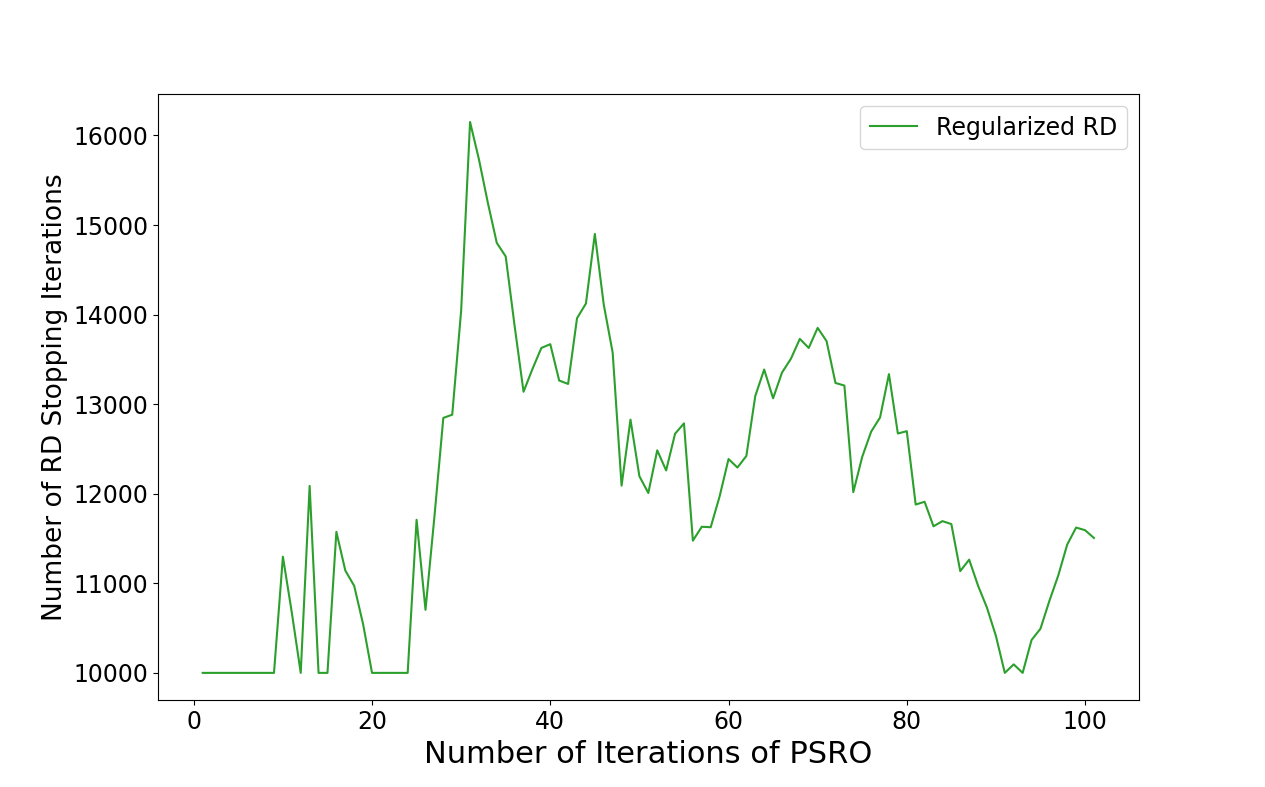}
\caption{The number of RD updates reaching the regret threshold $\lambda$ at different PSRO iterations.}\label{fig:RRD iter}
\end{figure}



\subsection{Game Descriptions} \label{app:std of hex}

\subsubsection{Real-world Game Instances}
The real-world games in our experiments are game models distilled from the full games by \citet{czarnecki2020real}. 
The descriptions of the games are shown as below.

\begin{itemize}
    \item Hex is a two-player board game, in which players attempt to connect opposite sides of a hexagonal board.
    \item Connect four is a two-player connection board game, in which the players choose a color and then take turns dropping colored tokens into the board. 
The objective of the game is to be the first to form a horizontal, vertical, or diagonal line of four of one's own tokens.
\item Misere Tic Tac Toe is a variant of Tic Tac Toe where one wins if and only if the opponent makes a line rather than itself.
\item Go with size 4 is a Go game with a smaller board than the original Go.
\end{itemize}

\subsubsection{Experimental Details for Real-World Games}

We evaluate all MSSs following the consistency metric by \citet{wang2022evaluating}, which measures the quality of intermediate game models. 
In all experiments, each player is initialized a set of poor-performing strategies at the beginning.
The set of initial strategies is identical for all MSSs.
Since these games are represented as matrix games, both best response computation and regret computation are deterministic.
Hence no error bar is reported.
For RRD, we set a regret threshold $\lambda=0.5$ for all experiments.

\subsubsection{Experimental Details for Bargaining Games}

We follow the bargaining game model by \citet{lewis2017deal}. 
Each player is given a uniformly generated value function, which gives a non-negative value for each item. 
The value functions are constrained so that: (1) the total value for a user of all items is 10; (2) each item has non-zero value to at least one user; and (3) some items have non-zero value to both players. 
These constraints enforce that it is not possible for both players to receive a maximum score, and that no item is worthless to both players, so the negotiation will be competitive. 
After 10 turns, players are able to complete the negotiation with no agreement, which is worth 0 points to both players.
For RRD, we set a regret threshold $\lambda=0.3$ for all experiments.

\subsubsection{Experimental Details for Attack-graph Games}

We follow the attack graph model of \citet{miehling2015planning}.
As described in the main paper, an attack graph is given by a DAG $\mathcal{G} = (V, E)$, where vertices $v \in V$ represent security conditions and edges $e \in E$ are labelled by exploits.
An \emph{attack-graph game} is defined by an attack graph endowed with additional specifications.
The game is a two-player partially observable stochastic game that is played over a finite number of time steps $T$\@.
At each time step $t$, the state of the game is given by the state of the graph, which is simply whether nodes are active or not (i.e., compromised by attacker or not), indicated by $s_{t}(v) \in \{0, 1\}$. 
It is assumed to be fully observable for the attacker while the defender receives a noisy observation $o_{t}(v) \in \{0, 1\}$ of the state, based on commonly known probabilities $P_v(o=0 \mid s=1)$ and $P_v(o=1 \mid s=0)$ for each node $v$.
Positive observations are called \emph{alerts}.

The attacker and defender act simultaneously at each time step~$t$. 
The defender's atomic action is to defend a node, thus, the atomic actions are simply~$V$. 
The defender's action space for any time step is $2^V$, meaning it can choose to defend any subset of the nodes. 
The attacker's atomic action set varies with time and is based on current graph state. 
Exploits on an edge are feasible only if the origin node of the edge is activated. Nodes without parents, called root nodes, can be attacked without preconditions. 
For the special case of attacking an \emph{AND} node, the attacker's atomic action is treated as attacking a node rather than choosing exploits on all incoming edges.
Thus, the attacker's atomic actions can be viewed as selecting edges (feasible exploits) for attacking \emph{OR} nodes or selecting nodes from \emph{AND} nodes whose parent nodes are all active. The attacker's action space at any time step
is the power set of feasible atomic actions.

Defender actions override attacker actions, that is, any node $v$ that is defended becomes inactive.
Otherwise, active nodes remain active; an \emph{AND} node $v$ that is attacked becomes active with probability $P(v)$, and any \emph{OR} node becomes active based on the success probabilities $P(e)$ of attacked edges.

Each \emph{goal} node, $v$, carries reward $R_A(v)$ for attacker and penalty $P_D(v)$ for defender for all time steps in which $v$ is active.
Any atomic action $a$ of an agent has a cost: $c_{a,D}(v)$ for nodes defended in case of defender; $c_{a,A}(v)$ for \emph{AND} nodes selection and $c_{a,A}(e)$ for edges selection in case of attacker. For simplicity, we omit the argument ($v$ or $e$) for action $a$ in the notation. 
When obvious from context we also drop the subscript $D$ and $A$, simply using $v_a$ and $c_a$ to denote the target node of $a$ and the cost of action $a$ respectively.

The defender's loss (negative reward) at any time step is the cost of all its atomic actions (i.e., total cost of nodes defended), plus the penalty for goal nodes active after the moves.
The defender's long-term payoff is the discounted expected sum of losses over time.
Similarly, the attacker's long-term payoff is the discounted expected sum of payoff per time step, where the per time step payoff is the reward for active goal nodes minus the cost of atomic actions used in that time step.

A policy for either player (pure strategy in the game) maps its observations at any step to a set of actions. 
For the attacker, the mapping is from states to action sets.
The defender only partially observes state, so its policy maps observation histories to action sets.
In our implementation, we limit the defender to a fixed length $h$ of past observations for tractability.
Solving the game means finding a pair of mixed strategies (distribution over pure strategies), one for each player, that constitutes a NE.

\subsection{Standard Deviations of Social Welfare in Bargaining Games}\label{app:SD bargain}
In Table~\ref{tab: std bargain}, we list the standard deviation of social welfare for difference solutions. 
This corresponds to Figure~\ref{fig: bargaining} in the main paper.
\begin{table}[!htb]
    \centering
    \begin{tabular}{lccccc}
        \toprule
        \textbf{MSSs} & Max SW & NE & Uniform & MWCCE & MGCCE\\
        \midrule
        RRD & 0.60 & 1.14 & 1.40 & 0.73 & 1.05 \\
        DO & 0.73 & 1.11 & 1.39 & 0.77 & 0.98 \\
        FP & 0.80 & 1.58 & 1.63 & 1.01 & 1.53 \\
        MWCCE & 0.67 & 1.11 & 1.66 & 0.70 & 2.49 \\
        MGCCE & 1.00 & 1.78 & 1.60 &1.11 &1.88 \\
        \bottomrule
    \end{tabular}
    \caption{Table of standard deviations for bargaining games.}\label{tab: std bargain}
\end{table}

\subsection{Decreased Regret by Regularization in PSRO with NE}\label{app:decreased regret}
In Figure~\ref{fig: PSRO with NE RRD vs NE}, we run PSRO with NE in two-player Leduc poker and plot the regret of the best-response targets (w.r.t the true game) given by RRD and the regret of NE (w.r.t the true game) in the intermediate empirical games. 
From Figure~\ref{fig: PSRO with NE RRD vs NE}, we observe that as in PSRO with RRD (Figure~\ref{fig: RRD NEvsRRD}), given an intermediate empirical game, the regret of best-response targets again decreases after applying regularization.

\begin{figure}[!htpb]
\centering
\includegraphics[width=0.95\columnwidth]{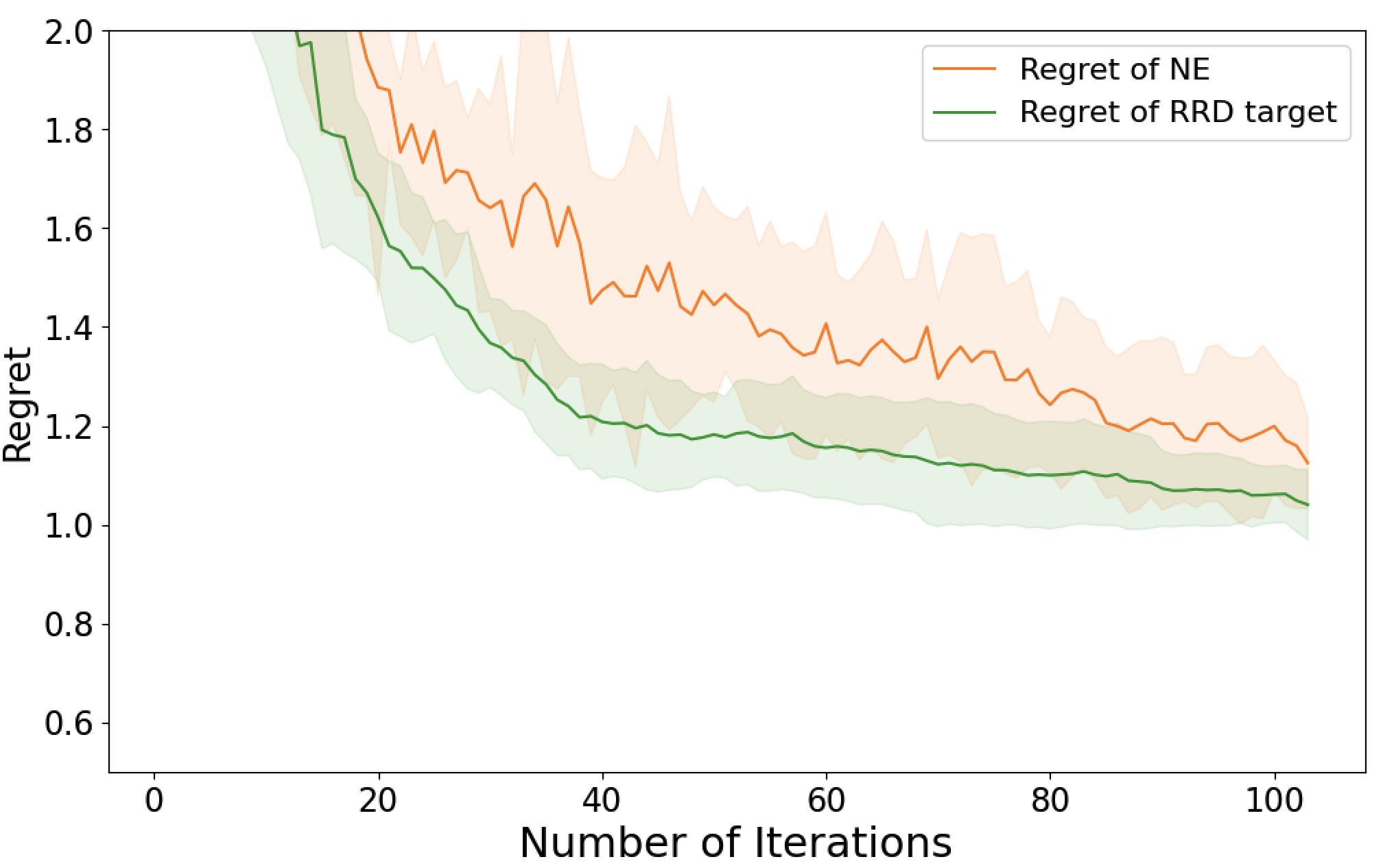}
\caption{Decreased regret when regularization is applied in PSRO run with NE as MSS.}\label{fig: PSRO with NE RRD vs NE}
\end{figure}

\subsection{Performance Check with Consistency Evaluation Solver}

In Figure~\ref{fig:regularization Leduc consistent}, we follow the rule of consistency \citep{wang2022evaluating}, a critical measure for the quality of generated empirical games. 
According to the consistency criterion, we compare MSSs with the same RRD-based regret and authenticate the faster improvement of RRD over DO and PRD in terms of the quality of generated empirical games.

\begin{figure}[!htpb]
\centering
\includegraphics[width=0.95\columnwidth]{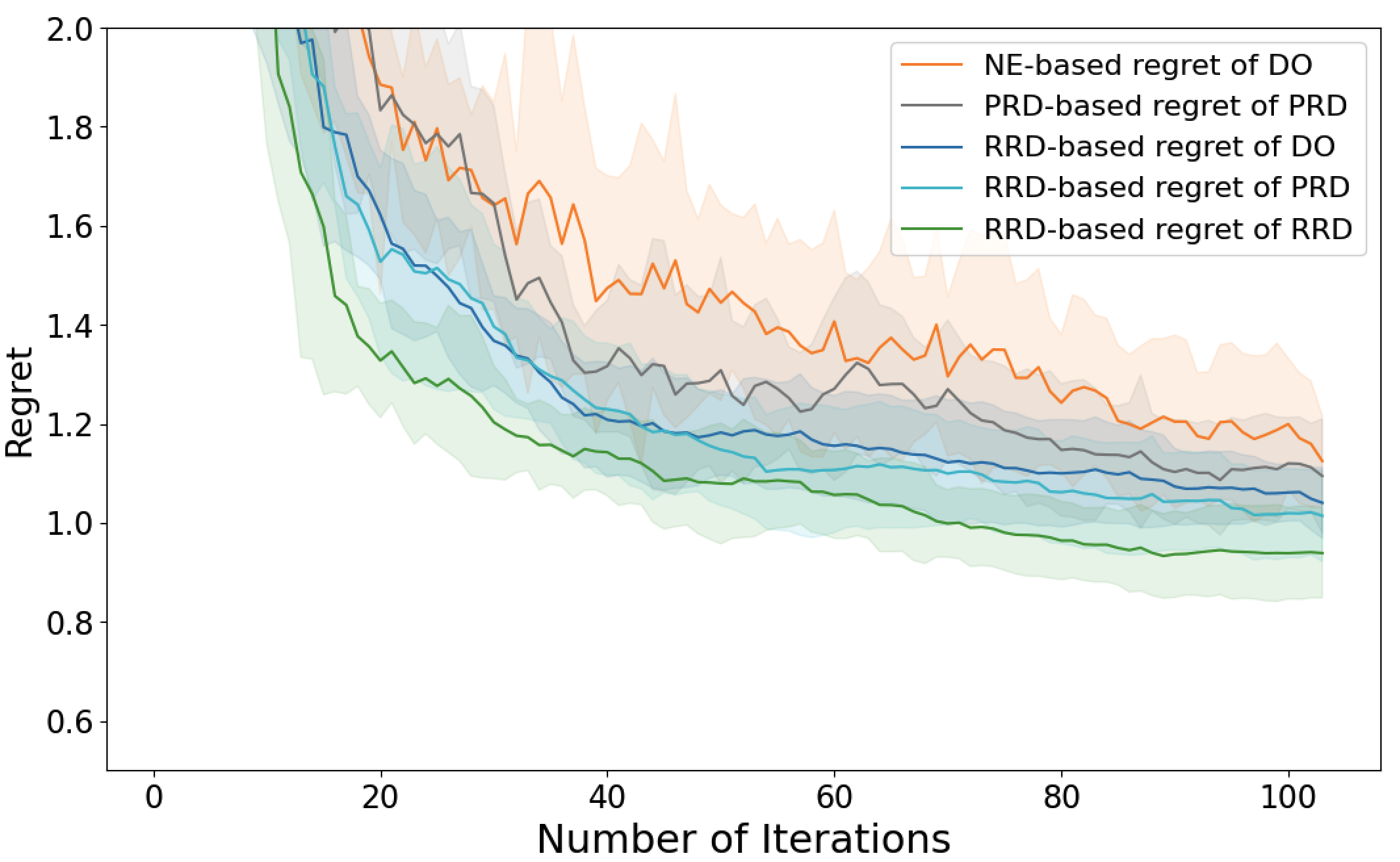}
\caption{RRD performance in two-layer Leduc poker with a consistency measure by \protect\citet{wang2022evaluating}.}\label{fig:regularization Leduc consistent}
\end{figure}

\subsection{Experimental Parameters} \label{app: hyperparams}
We use OpenSpiel \citet{lanctot2019openspiel} default parameter sets for experiments on Leduc: each payoff entry in an empirical game is an average of 1000 repeated simulations; DQN is adopted as a best response oracle, its parameters are shown in Table~\ref{tab:dqn_parameter}. 
The poker games are asymmetric in the sense that one player always moves first.

\begin{table}[!htb]
    \centering
    \begin{tabular}{lc}
        \toprule
        \textbf{Parameter} & \textbf{Value}\\
        \midrule
        learning rate & 1e-2\\
        Batch Size & 32\\
        Replay Buffer Size & 1e4\\
        Episodes & 1e4\\
        optimizer & adam\\
        layer size & 256\\
        number of layer & 4\\
        Epsilon Start & 1\\
        Epsilon End & 0.1\\
        Exploration Decay Duration & 3e6\\
        discount factor & 0.999\\
        network update step & 10\\
        target network update steps & 500\\
        \bottomrule
    \end{tabular}
    \caption{DQN hyperparameters}
    \label{tab:dqn_parameter}
\end{table}

PRD is implemented with lower bound for strategy probability 1e-10, maximum number of steps 1e5 and step size 1e-3. RD shares the same step size but a varying number of steps controlled by the regret threshold $\lambda$. We test the learning performance of RRD with $\lambda$ ranging from 0 to 0.6. We get best learning performance with $\lambda=0.35$ in Leduc poker. In three-player Leduc poker, we experiment with $\lambda=0.6$.

\begin{table}[!ht]
  \centering
\begin{tabular}{lc}
  \toprule
  \#Nodes & 100 \\
  Costs & Uniform in [0, 1] for attacker; \\
  &Uniform in [2, 4] for defender\\
  Rewards & Uniform in [10, 20] \\
  Penalties & Uniform in [7, 10]\\
  \#Goal Nodes & 6 \\
  Activation Prob. & Uniform in [0.6, 1] \\ 
  False Alarm Prob. &  Uniform in [0, 0.2]\\
  \bottomrule
\end{tabular}
\caption{Attack-graph game hyperparameters.} \label{tab:RGParams1}
\end{table}

\subsection{Code Availability}
Following the double-blind policy, we will publish our code after the review.

\subsection{Computational Resources}
We used one $2\times3.0$ GHz Intel Xeon Gold 6154 CPU with 16gb memories.

\section{PSRO with Other Novel MSSs}

\subsection{Results for QRE being an MSS}
One common assumption in game-theoretic analysis is the rationality of players (i.e., players act according to NE). 
Since our regularization approach prevents players from playing NE to some extent within the empirical game, it can be viewed as a way of restricting the rationality of players, which naturally relates our approach to Quantal Response Equilibrium (QRE) \citep{mckelvey1995quantal,mckelvey1998quantal}, an equilibrium notion with bounded rationality.
One common specification for QRE is logit equilibrium in which players' strategies take the following form 
\begin{displaymath}
    \sigma_i(s_i) = \frac{exp(\tau u_i(s_i, \sigma_{-i}))}{\sum_{s_i'\in S_i} exp(\tau u_i(s_i', \sigma_{-i}))}, \forall s_i\in S_i, i\in N.
\end{displaymath}
where $\tau$ is a parameter governing the rationality of players.

To see the performance of QRE being an MSS, we evaluate the learning performance of PSRO with QRE. 
Specifically, we compute the QRE of the empirical game at every PSRO iteration using Gambit \citep{mckelvey2006gambit} and analyze the learning performance of QRE.

\begin{figure}[!htpb]
\centering
\includegraphics[width=\columnwidth]{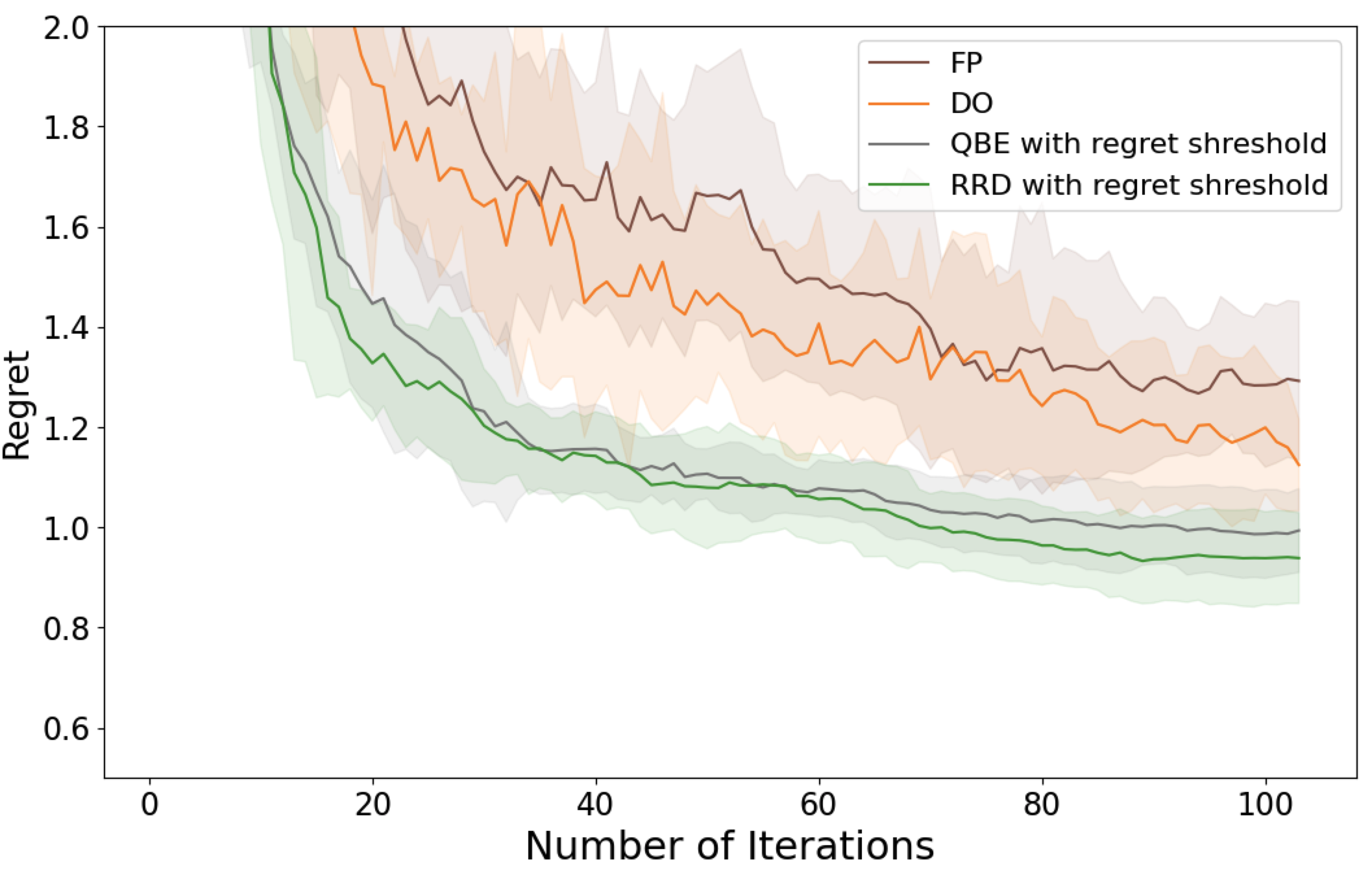}
\caption{Performance of PRSO with QRE in two-layer Leduc poker.}\label{fig: QRE}
\end{figure}

Figure~\ref{fig: QRE} shows the learning performance in Leduc poker with QRE as an MSS. 
For comparison, we also plot the learning curve of RRD with the same regret threshold of QRE.
Although QRE shows a slight divergence in the end, it still demonstrates the potential of using QRE as an MSS in PSRO.

\subsection{Results for Using MRCP as an MSS}\label{app: MRCP as MSS}
\subsubsection{Definition of MRCP}
The profile in the empirical game closest to being a solution of the full game is called \newterm{minimum-regret constrained-profile} (MRCP) \citep{Jordan10sw}.
Formally, $\bar{\sigma}$ is an MRCP iff: 
\begin{displaymath}
     \bar{\sigma} = \argmin_{\sigma\in \Delta(X)} \sum_{i\in N} \rho^{\mathcal{G}}_i(\sigma)
\end{displaymath}
The regret of MRCP thus provides a natural measure of how well $X$ covers the strategically relevant space \citep{wang2022evaluating}.

\subsubsection{Experiments of Learning with MRCP}
We have observed the existence of strategy profiles with lower global regret than NE in the empirical game and the experimental results of regularization shows that training against them results in improved learning performance than DO. 
One natural question to ask is whether training against the most stable profile targets can benefit strategy exploration the most (e.g., using MRCP as MSS). 

To answer this question, we compare the performance of MRCP as MSS against DO and FP in the matrix-form two-player Kuhn's poker and a synthetic two-player zero-sum game. 
In Kuhn's poker, we randomly select 4 starting points and implement PSRO. 
Fig.~\ref{fig:Kuhn's Poker1}-\ref{fig:Kuhn's Poker4} show that with 3 out of 4 starting points, MRCP converges slight faster than DO. 
For the matrix game, Fig.~\ref{fig:MRCP_zs1} and~\ref{fig:MRCP_zs2} show the benefits of applying MRCP but the performance varies across different starting points.

In Fig.~\ref{fig: MRCP as an MSS}, we observe that the MRCP has some power for heuristic strategy generation. 
However, the advantage of using MRCP is not satisfactory in terms of convergence rate and computational complexity.
We also find that using MRCP may converge slower in other games like Blotto compared to DO and PRD. 

\begin{figure*}[!htbp]
\centering
    \subfloat[Kuhn's Poker]{\includegraphics[height=3.2cm,width=0.32\linewidth]{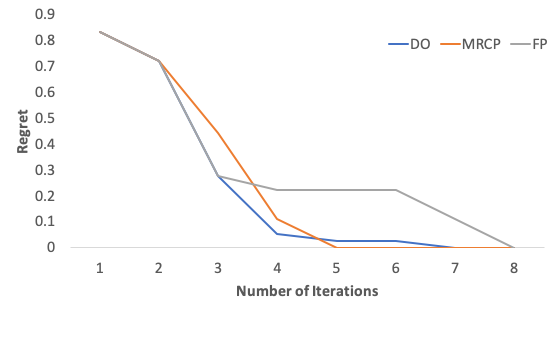}\label{fig:Kuhn's Poker1}}\ 
    \subfloat[Kuhn's Poker]{\includegraphics[height=3.2cm,width=0.32\linewidth]{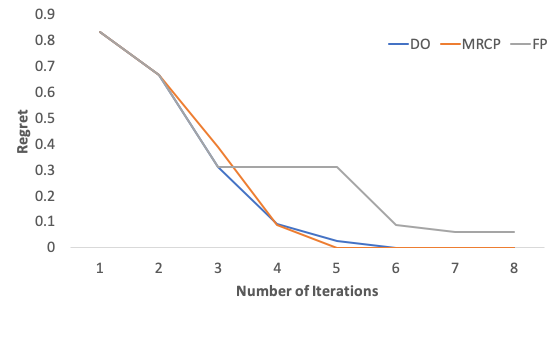}\label{fig:Kuhn's Poker2}}\
    \subfloat[Kuhn's Poker]{\includegraphics[height=3.2cm,width=0.32\linewidth]{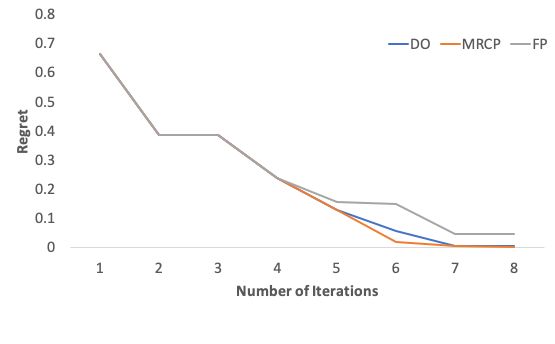}\label{fig:Kuhn's Poker3}}\\
    \subfloat[Kuhn's Poker]{\includegraphics[height=3.2cm,width=0.32\linewidth]{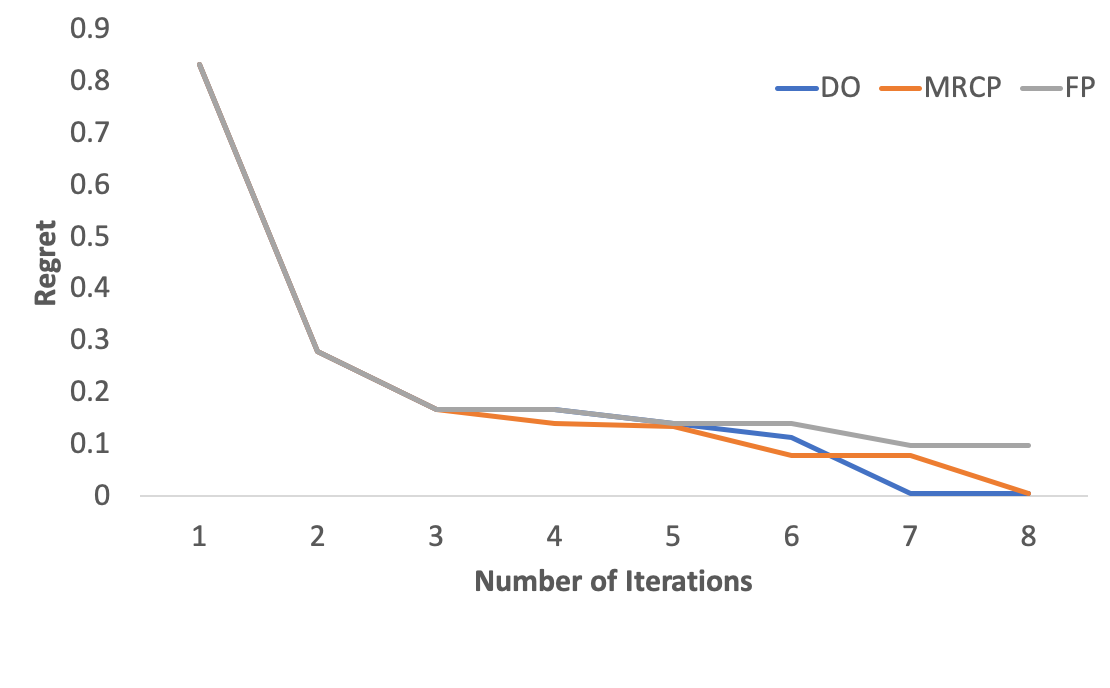}\label{fig:Kuhn's Poker4}}\
    \subfloat[Zero-Sum Game ]{\includegraphics[height=3.2cm,width=0.32\linewidth]{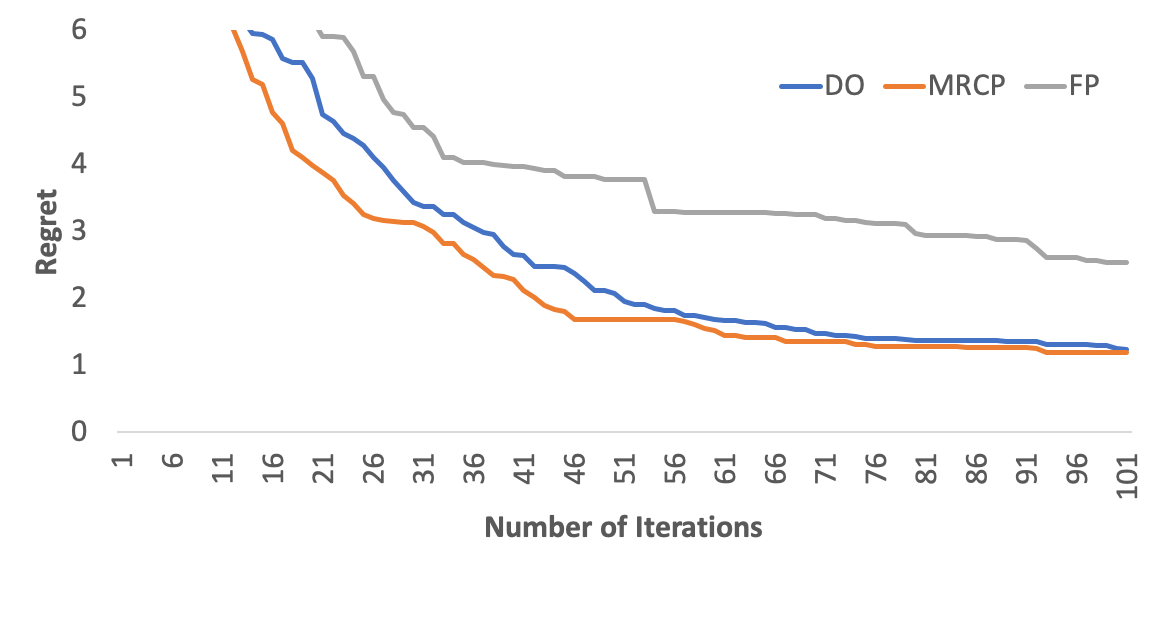}\label{fig:MRCP_zs1}}\
    \subfloat[Zero-Sum Game]{\includegraphics[height=3.2cm, width=0.32\linewidth]{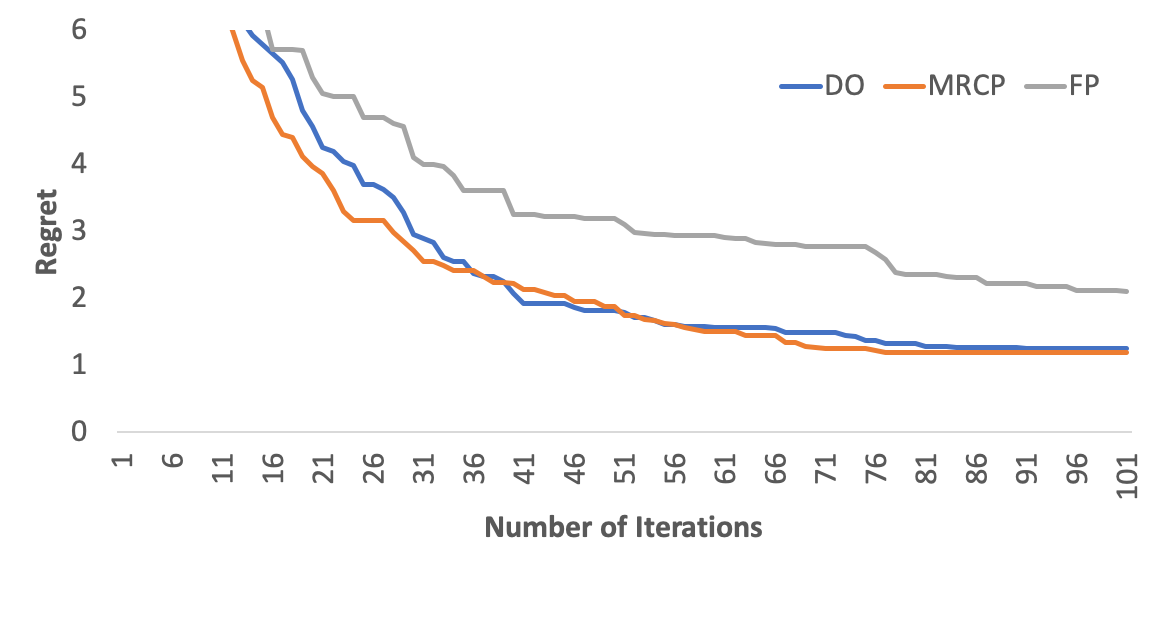}\label{fig:MRCP_zs2}}
\caption{Performance of using MRCP as an MSS. Y axis depicts MRCP-based regret.}\label{fig: MRCP as an MSS}
\end{figure*}

The experiments show that training against the lowest-regret profile in the empirical game does not necessarily lead to a better overall learning performance.
This is because the lowest-regret profile in the empirical game may not be changed a lot after adding a new strategy to the empirical game, which yields similar strategies continued to be added over PSRO iterations.
Continuing adding similar strategies would result in only little performance improvement over PSRO iterations.


Now we illustrate why pursuing best-response targets with extremely low regret may result in slow learning using a matrix game shown in Table~\ref{tab: game with long NE path}.
The matrix game contains 1000 strategies for each player. 
All missing entries of the payoff matrix are $(0,0)$. 
Let's start PSRO with the first strategy $(s_1, s_1)$. 
This matrix game is designed to have long equilibrium search path for DO (as in many real-world games). 
Specifically, by best-responding to $(s_1, s_1)$, each player adds $s2$ to the empirical game whose new NE is $(s_2, s_2)$. 
Similarly, if we best-respond to the equilibrium at each PSRO iteration, we would first get a new NE $(s_3, s_3)$ and then a long equilibrium path through the diagonal until we reach the NE of the full game $(s_{1000}, s_{1000})$.

Without loss of generality, suppose we are at iteration 2 (i.e., the empirical game includes $(s_1, s_2)$ and $(s_2, s_2)$ is an empirical NE). 
The MRCP of this empirical (symmetric) game is approximately $(1 s1, 0 s2)$ with regret 0.0112=0.022 (sum over players) (the regret of accurate MRCP is even lower). 
The regret of empirical NE is $0.1*2=0.2$ by deviating to $s_3$ from $(s_2, s_2)$. 
When best responding to MRCP, we add $s_{500}$ (only considering deviation strategies outside the empirical) and then MRCP remains the same (i.e., $(1 s_1, 0 s_2)$) for the empirical game. 
Therefore, further best responding to the MRCP may again add some strategies similar to $s_{500}$ and may not improve the learning performance dramatically.

Suppose RRD gives probability $(0.5, 0.5)$ on $(s_1, s_2)$, then best responding to $(0.5s_1, 0.5s_2)$ leads to equilibrium strategy $s_{1000}$ directly, jumping out of the long equilibrium path of DO. The regret of $(0.5s_1, 0.5s_2)$ is $(0.005 \times 0.5+0.199\times0.5 - 0.011\times0.25 - 0.1*0.25)\times2 = (0.102 - 0.02775)\times2 = 0.074252 = 0.1485$ (0.02 (regret of MRCP) $<$ 0.1485 (regret of RRD) $<$ 0.2 (regret of NE)).

In this example, best responding to the relatively low full-game regret profile, RRD avoids falling into the long diagonal path as DO. 
Meanwhile, its regret is not as low as the regret of MRCP so that the best-response target at each PSRO iteration would keep being updated significantly rather than staying similarly. 


\subsubsection{Extra Properties of Learning with MRCP}
Theoretically, multiple MRCPs could exist in an empirical game and MRCP is not necessarily a pure strategy profile in general. 
Moreover, purely using MRCP as an MSS does not guarantee convergence to NE since the best-responding strategy could already exist in the empirical game. We define this property of MRCP as follows. 

\theoremstyle{definition}
\begin{definition*}
An empirical game with strategy space $X\subseteq S$ is \newterm{closed} with respect to MRCP $\bar{\sigma}$ if
\begin{displaymath}
     \forall i \in N, s_i = \argmax_{s_i'\in S_i}u_i(s_i', \sigma_{-i}) \in X_i.
\end{displaymath}
\end{definition*}

To illustrate this concept, consider the symmetric zero-sum matrix game in Table~\ref{tab: MRCP closed}. Starting from the first strategy of each player and implementing PSRO with MRCP, we have the empirical game including $a^1$ and $a^2$. Since the $(a^1_1, a^1_2)$ is a MRCP (considered all pure and mixed strategy profiles) and best responding to the profile gives $a^2$ again, the empirical game is \textit{closed} and never extends to the true game wherein the true NE is $(a^3_1, a^3_2)$. 

\begin{table}[!ht]
    \centering
    \begin{tabular}{ |c|c|c|c| } 
 \hline
  & $a_2^1$ & $a_2^2$ & $a_2^3$  \\ 
 \hline
 $a_1^1$ & (0, 0) \textcolor{orange}{[2]} & (-1, 1)  \textcolor{orange}{[6]} & (-0.5, 0.5) \\ 
 \hline
 $a_1^2$ & (1, -1)  \textcolor{orange}{[6]} & (0, 0)  \textcolor{orange}{[10]} & (-5, 5) \\ 
 \hline
 $a_1^3$ & (0.5, -0.5) & (5, -5) & (0, 0) \\ 
 \hline
\end{tabular}
\caption{Symmetric Zero-Sum Game for MRCP. Regret of profiles is shown in the square parenthesis.}
\label{tab: MRCP closed}
\end{table}{}

In our experiments, we deal with this issue by only introducing new strategy with highest deviation payoff outside the empirical game and thus guarantee convergence. 
An alternative is to switch between DO and MRCP whenever this issue happens and the convergence is guaranteed due to the convergence property of DO.

\section{Convergence of RRD}\label{app: convergence}
\subsection{Proof of Convergence}
We first define the concept \newterm{$\epsilon$-closeness}, which can be viewed as a stopping condition of PSRO. 

\begin{definition*}[$\epsilon$-closeness]
An empirical game with strategy space $X\subseteq S$ is \newterm{$\epsilon$-closed} with respect to certain $\epsilon$-NE $\sigma \in \Delta(X)$ and operator $o$ if and only if $o(\sigma) \in X$. 
\end{definition*}

For example, if $o$ is a best-response operator and $\epsilon=0$, this definition means there is no beneficial deviation from the NE $\sigma$ of the empirical game, and thus $\sigma$ is a NE of the full game. When $\epsilon \ne 0$, $\epsilon$-closeness indicates that the deviation strategy of the $\epsilon$-NE $\sigma$ of the empirical game already exists in the empirical game.
Note that there could exist an infinite number of $\epsilon$-NE in an empirical game given a specific $\epsilon$, so the definition of $\epsilon$-closeness is associated with a specific $\epsilon$-NE.

Next we prove that if an empirical game is $\epsilon$-closed with respect to certain $\epsilon$-NE $\sigma \in \Delta(X)$ and best-response operator $o$, then $\sigma$ is an $\epsilon$-NE of the full game.

\begin{lemma} \label{lemma: closeness}
If an empirical game with strategy space $X\subseteq S$ is $\epsilon$-closed with respect to certain $\epsilon$-NE $\sigma \in \Delta(X)$ and best-response operator $o$, then $\sigma$ is an $\epsilon$-NE of the full game $\mathcal{G}$. 
\end{lemma}
\begin{proof}
Since $\sigma$ is an $\epsilon$-NE in the empirical game, there is no deviation strategy within the empirical game that results in regret large than $\epsilon$. 
Mathematically, we have $\forall i\in N$, $\max_{s_i'\in X_i}u_i(s_i', \sigma_{-i})-u_i(\sigma_i, \sigma_{-i}) \le \epsilon$. 
Since the best-response operator finds the best deviation w.r.t the true game and the best deviation falls into the empirical game, we have $\forall i\in N$, $\max_{s_i'\in S_i}u_i(s_i', \sigma_{-i})-u_i(\sigma_i, \sigma_{-i}) \le \epsilon$. 
Then $\sigma$ is an $\epsilon$-NE of the full game $\mathcal{G}$. 
\end{proof}

Consider the finite strategy space $S$. We prove the following theorem that if we train against an $\epsilon$-NE of the empirical game at each iteration of PSRO, we end up with an empirical game containing at least one $\epsilon$-NE.
By setting $\epsilon$ to be a reachable $\lambda$, we prove the Theorem~\ref{thm:CDO convergence}.

\begin{theorem}\label{thm:CDO convergence}
Assuming the access to an exact best response oracle, Policy Space Response Oracle with Regularized Replicator Dynamics associated with a reachable regret threshold $\lambda$ converges to an empirical game containing at least one $\lambda$-NE. 
\end{theorem}

\begin{proof}
Since we have finite strategy space $S$, $\epsilon$-closeness with respect to certain $\sigma$ is always reachable by training against an $\epsilon$-NE at each iteration. 
Once the $\epsilon$-closeness is reached, the corresponding $\sigma$ is an $\epsilon$-NE of the full game due to Lemma~\ref{lemma: closeness}. 
Due to the population property of PSRO \citep{wang2022evaluating}, profiles with lower regret than $\epsilon$ could also exist. 
\end{proof}
This result generalizes DO and its convergence guarantee to scenarios where training target is not strictly restricted to NE\@. 
Note that similar results have been proved in earlier works \citep{Dinh22,mcaleer2021cfr}. 
However, they miss the fundamental fact of an empirical game, that is, an empirical game creates a profile space, within which profiles with regret much smaller than the target profile could exist. 
This fact reveals one major advantage of using game models to facilitate game learning since with a game model, we simply need to capture a NE within the strategy space rather than requiring a sequence of profiles converge to NE as in the online learning setting.

\subsection{Performance Confusion}
A common confusion is why RRD is useful in the sense that RRD has weaker theoretical convergence guarantee than DO. 
Note that the proof for the convergence of DO by \citet{mcmahan2003planning} describes whether it converges rather than how fast it converges. 
In the worst case, DO needs to add all strategies (assuming the game is finite) to the empirical game to reach the convergence, which is apparently trivial.
In games of interest (e.g., Leduc poker), the strategy space is usually very large and cannot be enumerated.
Despite algorithms like DO guarantee a convergence to NE, it is possible that such algorithms need to explore a large number of strategies to reach the convergence, which is far beyond the computational budget.
The goal of strategy exploration is to learn towards a NE as close as possible within the budget.
Therefore, strategy exploration focuses more on the learning performance of MSSs within the computational budget and care less about the long-term convergence beyond the budget.
Although RRD has weaker convergence guarantee in the limit, it is able to learn stable profiles quickly within the budget by controlling the degree of regularization, which fulfills the spirit of strategy exploration.
If an absolute convergence is pursued, one can always switch RRD to DO at certain PSRO iteration after useful strategies have been quickly generated by RRD.

\begin{table}[!ht]
    \centering
    \begin{tabular}{ |c|c|c|c|c|c|c|c| } 
 \hline
  & $a_2^1$ & $a_2^2$ & $a_2^3$ & ... & $a_2^{500}$ & ... & $a_2^{1000}$\\ 
 \hline
 $a_1^1$ & (0, 0) & (0, 0.011) & (0, 0) & ... & (0. 0.01)& ...& (0,0.005) \\ 
 \hline
 $a_1^2$ & (0.011, 0) & (0.1, 0.1) & (0.1, 0.2) & ... & ...& ... & (0,0.199) \\ 
 \hline
 $a_1^3$ & (0, 0) & (0.2, 0.1) & (0.2, 0.2) & ... & ...& ... & (0, 0)\\ 
 \hline
 ... & ... & ... & ... & ... & ... & ...& ...\\ 
 \hline
 $a_1^{500}$ & (0.01, 0) & ... & ... & ... & ...& ... & (0, 0) \\ 
 \hline
 ... & ... & ... & ... & ... & ... & ...& ...\\ 
 \hline
 $a_1^{1000}$ & (0.005, 0) & (0.199, 0) & (0, 0) & ... & (0, 0)& ... & (100, 100) \\ 
 \hline
\end{tabular}
\caption{A game instance with a long NE path.}
\label{tab: game with long NE path}
\end{table}{}

\end{document}